\documentclass[sigconf]{acmart}
\copyrightyear{2023}
\acmYear{2023}
\setcopyright{acmlicensed}
\acmConference[CCS '23] {Proceedings of the 2023 ACM SIGSAC Conference on Computer and Communications Security}{November 26--30, 2023}{Copenhagen, Denmark.}
\acmBooktitle{Proceedings of the 2023 ACM SIGSAC Conference on Computer and Communications Security (CCS '23), November 26--30, 2023, Copenhagen, Denmark}
\acmPrice{15.00}
\acmISBN{979-8-4007-0050-7/23/11}
\acmDOI{10.1145/3576915.3616643}

\usepackage{amsthm, amsmath,amssymb,amsfonts}
\usepackage{algorithmic}
\usepackage{graphicx}
\usepackage{textcomp}
\usepackage{xcolor}
\usepackage{color, colortbl} 
\usepackage{multirow}
\usepackage{enumitem}
\usepackage{multicol}
\usepackage[altpo]{backnaur}
\usepackage{booktabs} 
\usepackage{listings}
\usepackage{tikz}
\usepackage{float}
\usepackage{url}
\usepackage{hyperref}
\usepackage{mathpartir}
\usepackage{fancyhdr}
\usepackage{balance}
\def\BibTeX{{\rm B\kern-.05em{\sc i\kern-.025em b}\kern-.08em
    T\kern-.1667em\lower.7ex\hbox{E}\kern-.125emX}}

\newif\ifdraft
\draftfalse

\newif\ifhighlightrev
\highlightrevtrue
\newcommand{\rev}[1]{\ifhighlightrev{\color{black}#1}\fi}

\newcommand{\smallStep}[2]{\langle #1, #2 \rangle}

\newcommand{\sn}[5]{\smallStep{#1}{#2} \xrightarrow{t(#2)} \langle \mathsf{#3}, #4 \rangle}
\newcommand{\snr}[3]{\smallStep{#1}{#2}\rightarrow_r   #3 }

\def\G{\Gamma}
\def\pu{public}
\def\pud{public~prog}
\def\pr{private}

\def\pum{\mathsf{\pu}}

\newcommand{\tp}[2]{\G \vdash #1 : \mathsf{#2}}

\theoremstyle{definition}
\newtheorem{definition}{Definition}[section]
\newtheorem{theorem}{Theorem}[section]
\newtheorem{lemma}[theorem]{Lemma}

\newcommand{\stress}[1]{\textbf{\textit{#1}}}

\newcommand{\fullcirc}[1]{
    \begin{tikzpicture}
    \filldraw [black,fill=#1] (0,0) circle (0.8ex); 
    \end{tikzpicture}
}
\newcommand{\dfull}[0]{\fullcirc{darkgray}}

\newcommand{\dnone}[0]{\fullcirc{white}}

\begin{document}
\title{Security Verification of Low-Trust Architectures} 
%


\author{Qinhan Tan}
\authornote{Three co-first authors contributed equally to this research.}
\email{qinhant@princeton.edu}
\orcid{0000-0003-2475-3675}
\affiliation{%
  \institution{Princeton University}
  \city{Princeton}
  \state{New Jersey}
  \country{USA}
}


\author{Yonathan Fisseha}
\authornotemark[1]
\email{yonathan@umich.edu}
\orcid{0009-0000-9645-2885}
\affiliation{%
  \institution{University of Michigan}
  \city{Ann Arbor}
  \state{Michigan}
  \country{USA}
  \postcode{48109}
}

\author{Shibo Chen}
\authornotemark[1]
\email{chshibo@umich.edu}
\orcid{0000-0002-9522-8934}
\affiliation{%
  \institution{University of Michigan}
  \city{Ann Arbor}
  \state{Michigan}
  \country{USA}
  \postcode{48109}
}

\author{Lauren Biernacki}
\email{biernacl@lafayette.edu}
\orcid{0000-0001-8511-2287} 
\affiliation{%
  \institution{Lafayette College}
  \city{Easton}
  \state{Pennsylvania}
  \country{USA}
}

\author{Jean-Baptiste Jeannin}
\email{jeannin@umich.edu}
\orcid{0000-0001-6378-1447}
\affiliation{%
  \institution{University of Michigan}
  \city{Ann Arbor}
  \state{Michigan}
  \country{USA}
}

\author{Sharad Malik}
\email{sharad@princeton.edu}
\orcid{0000-0002-0837-5443}
\affiliation{%
  \institution{Princeton University}
  \city{Princeton}
  \state{New Jersey}
  \country{USA}
}

\author{Todd Austin}
\email{austin@umich.edu}
\orcid{0000-0002-0181-0852}
\affiliation{%
  \institution{University of Michigan}
  \city{Ann Arbor}
  \state{Michigan}
  \country{USA}
}

\begin{abstract} 

Low-trust architectures work on, from the viewpoint of software, always-encrypted data, and significantly reduce the amount of hardware trust to a small software-free enclave component.
In this paper, we perform a complete formal verification of a specific low-trust architecture, the Sequestered Encryption (SE) architecture, to show that the design is secure against direct data disclosures and digital side channels for all possible programs. We first define the security requirements of the ISA of SE low-trust  architecture. Looking upwards, this ISA serves as an abstraction of the hardware for the software, and is used to show how any program comprising these instructions cannot leak information, including through digital side channels. Looking downwards this ISA is a specification for the hardware, and is used to define the proof obligations for any RTL implementation arising from the ISA-level security requirements. These cover both functional and digital side-channel leakage. Next, we show how these proof obligations can be successfully discharged using commercial formal verification tools. We demonstrate the efficacy of our RTL security verification technique for seven different correct and buggy implementations of the SE architecture. 
\end{abstract}

\ccsdesc[500]{Security and privacy~Formal methods and theory of security}
\ccsdesc[500]{Security and privacy~Information flow control}
\ccsdesc[500]{Security and privacy~Security requirements}
\ccsdesc[300]{Security and privacy~Side-channel analysis and countermeasures}

\keywords{Low-Trust Architecture, Information Flow, Formal Verification}
  


\maketitle

{
\let\spacysection\section
\renewcommand{\section}[1]{\vspace{-0.35em}\spacysection{#1}\vspace{-0.2em}}
\let\spacysubsection\subsection
\renewcommand{\subsection}[1]{\vspace{-0.35em}\spacysubsection{#1}\vspace{-0.05em}}
\let\spacysubsubsection\subsubsection
\renewcommand{\subsubsection}[1]{\vspace{-0.2em}\spacysubsubsection{#1}\vspace{-0.1em}}

\section{Introduction}
Security verification of a computing system, while highly desirable, is a challenging task that often falls short of the desired level of guarantees. The verification must be applied to all trusted components of the system, including hardware and software. Unlike penetration testing, which consists of focused attempts to infiltrate a system, formal security verification is a proof that a particular security vulnerability does not exist within a design. Unfortunately, most systems today receive little to no formal security verification, due to design complexity challenges and limitations of formal proof mechanisms. Design complexity manifests in the sheer size of today's secure systems, which comprise architectures, microarchitectures, and deep software stacks, all of which must be trusted and verified. These complex systems easily exceed the capabilities of today's formal proof mechanisms, such as SAT/SMT solvers, model checkers, and proof assistants. Consequently, incomplete penetration testing still remains the backbone of today's security verification efforts.

\stress{Low-trust architectures} have recently emerged as a secure system design framework that {\em i)} eliminates all trust in software, and {\em ii)} significantly reduces the amount of hardware trust to a small, software-free enclave component. 
These properties make formal security verification feasible by shrinking the system aspects that must be trusted and verified. 
In this paper, we focus on the security verification of Sequestered Encryption (SE)~\cite{se-paper}—a low-trust architecture that claims to protect the confidentiality of sensitive data against direct data disclosures and digital side channels. \rev{Direct disclosures refer to any direct leakage of plaintext values through SE computation. Digital side channels represent any indirect leakage of plaintext values through non-analog information paths, including analysis of ciphertext values, operational timing, program control flow, memory access patterns, or microarchitectural resource usage. Currently, SE Enclave does not protect against analog information flow paths, such as frequency throttling~\cite{hertzbleed, frequency-throttling}, power analysis~\cite{kocher1999differential, goubin1999and,mangard2003simple}, electromagnetic snooping~\cite{oneanddone,dropanddrop,nscwithEM}, etc.}

Within SE, the instruction set architecture (ISA) consists of the native instructions, termed the native ISA, and a set of secure instructions termed the SE ISA. 
The native ISA contains insecure instructions to be executed by unsecured (or untrusted) components that do not have access to secret values. 
In contrast, the SE ISA consists of secure instructions that operate on encrypted data and are executed solely in its software-free enclave component.
It is the design of this SE Enclave and ISA extension that ensures the cryptographic-strength confidentiality of the sensitive data.


\begin{figure}[t]
    \centering
    \includegraphics[width=\columnwidth,trim={6cm 3cm 0.9cm 1.3cm},clip]{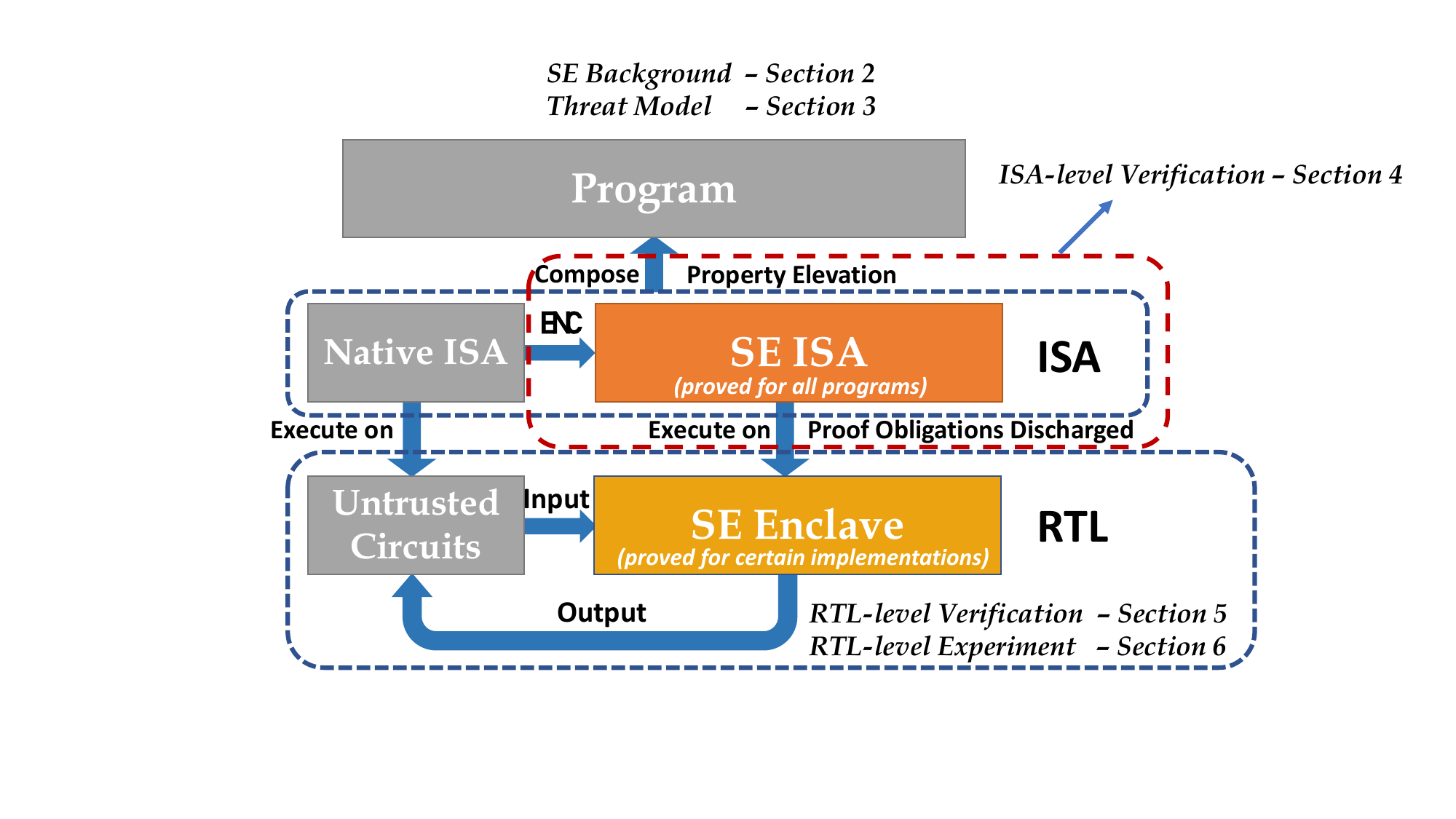}
    \vspace{-2em}
    \caption{Proof System and Paper Organization}
    \vspace{-2.2em}
    \label{fig:proof-system}
\end{figure}

In this work, we perform a complete formal verification of the SE low-trust architecture to show that the design is secure against direct data disclosures and digital side channels for all possible programs. The steps involved are illustrated in Figure~\ref{fig:proof-system}. First, we articulate a set of instruction-level security requirements that SE ISA must fulfill to prevent disclosures or digital side channels. Going upwards to the software, these are used in hand-driven proof techniques to show that no SE program can possibly create a disclosure or digital side-channel leakage. Going downwards toward the hardware, the ISA security requirements necessitate security properties to be enforced in the hardware design, typically called RTL-level security properties.
Last, we propose a verification scheme to formally verify RTL-level properties using an off-the-shelf commercial formal verification tool (Cadence's JasperGold~\cite{website:jaspergold}) on real designs. 
Our evaluation experiments show that our RTL-level verification scheme can prove the security properties being met by correct implementations, and also capture security leakages in flawed implementations.
This demonstrates the practical applicability of our proposed verification scheme to designs at scale.
To our knowledge, this is the first formal verification of a secure computing framework that extends to both direct disclosures and digital side channels, for all possible programs running on a verified computing platform.

A key takeaway from this work is that low-trust architectures lend themselves to formal security verification. We find two primary reasons for this outcome. First, the nature of low-trust architectures eliminates any trust in software. Since our verification ensures that the software can only see values encrypted under semantically secure cryptosystems~\cite{semnatically-secure}, software verification is not part of the overall proof system. Using only hand-driven proof, we show that any program using our ISA cannot disclose or create digital side channels, thus ending concern for any software. 
In traditional security verification, where properties must be proven partly in hardware and software, often the complexity of software reasoning leads to compromises on what can be guaranteed in these systems. 
Second, the simplicity of the low-trust SE hardware enclave, having minimal state and control logic, allows all of our ISA-level-based security assertions to complete on the actual RTL of the design, ensuring no gaps between the deployed design and any potential abstractions employed to enable formal verification.
We are confident that the approach we have detailed in this paper will extend itself to future low-trust architectures as they become available.

{\vspace{1em}
\noindent\textit{Contributions.}}
\rev{This work observes and demonstrates how low-trust architectures enable end-to-end software-to-hardware verification of strong security attributes, \textit{i.e.}, confidentiality and digital side-channel free execution.  Below is the list of specific contributions:
\vspace{-1em}
\begin{itemize}[leftmargin=*]
    \item \textbf{Formal SE ISA Semantics}: Define formal SE ISA semantics that enable privacy-related reasoning.
    \item \textbf{Software-Level Proof}: Povide proof of confidentiality and side-channel freedom for \textbf{\textit{all programs}} written using this ISA.
    \item \textbf{Hardware Proof Obligations}: Provide proof obligations for any hardware implementation of the SE ISA to serve as the interface between the verification of the software and the hardware.
    \item \textbf{Hardware-Level Proof}: 
    \begin{itemize}[leftmargin=*]
        \item Provide a list of security properties that meet the hardware proof obligations for a specific SE hardware implementation.
        \item Demonstrate that these properties can be checked using standard information flow tracking (IFT) and commercial off-the-shelf IFT tools with novel RTL-verification elements.
        \item Demonstrate how the checking detects bugs in four buggy different implementations that violate these properties.
    \end{itemize}
\end{itemize}
\vspace{-1em}

To further clarify the contributions of this work, we do not claim generalizability beyond the small enclaves in low-trust architectures detailed in \S\ref{sec:background}. In fact, the increased level of verification is enabled by the low-trust architecture's property of limiting trust to only within the small SE Enclave, which in turn eliminates all trust in software and significantly reduces the degree of hardware that must be trusted (and thus needs to be verified) to ensure the proof properties.}

Figure~\ref{fig:proof-system} also serves to illustrate our paper organization. In \S~\ref{sec:background}, we provide a brief overview of SE. In \S~\ref{sec:threat-model}, we articulate the threat model and the scope of this work. In \S~\ref{sec:isa}, we formalize ISA-level properties and prove the program-level properties inducted from the ISA-level properties followed by proof responsibilities discharged to the RTL-level. In \S~\ref{sec:definition}, we present our modeling and verification strategy at the RTL-level. In \S~\ref{sec:evaluate}, we describe the SE design variances and apply our verification scheme to these designs. We close the paper with related works (\S~\ref{sec:related}) and final conclusions (\S~\ref{sec:conclusions}).

Our designs and verification scripts are available at \url{https://github.com/qinhant/SE_verification_CCS}.


\section{Background}\label{sec:background}
In this section, we provide an overview of Sequestered Encryption (SE)~\cite{se-paper} and show that the unique characteristics and design principles of SE open up new opportunities for hardware verification to provide complete reasoning of the underlying computing paradigm without the knowledge of specific programs. 
\rev{A more detailed study of the SE design, including comparisons with related work, can be found in \cite{se-paper}.}


\subsection{Sequestered Encryption (SE)}

Sequestered Encryption (SE) is a hardware-based technique to protect the confidentiality of secret third-party data during computation. 
With SE, third-party data is encrypted in a trusted, client-side environment and offloaded to a server for computation. The server application operates on third-party data using SE's ISA instructions backed by custom hardware support. 
Specifically, SE extends the conventional ISA to include secure instructions that operate on ciphertext data. We call these additional secure instructions the SE ISA. These instructions are dispatched to SE's hardware enclave, which computes the requested operation on the source ciphertexts. 

\rev{The SE Enclave works to significantly reduce software and hardware trust by sequestering all sensitive computation to the hardware enclave. As such, the SE Enclave design claims that sensitive private data cannot be disclosed by any SE instruction sequence, either directly through a disclosure or indirectly through a digital side channel. This claim covers all digital side channels, including cryptanalysis of the ciphertext emitted by the Enclave, Enclave operational timing, and any possible influence the SE Enclave has on the system's memory access patterns and control flow. The goal of this paper is to formalize and prove these claims for the SE Enclave, through proofs on the SE instruction set operational semantics and their proof obligations expressed in the RTL implementation in the SE Enclave. \textbf{\textit{With these claims proven, the SE Enclave represents the first enclave that has been formally verified to not suffer from software vulnerabilities or digital side channels.}}}

\begin{figure}[ht]
    \centering
    \includegraphics[width=0.44\textwidth,trim={0cm 5cm 2cm 0cm},clip]{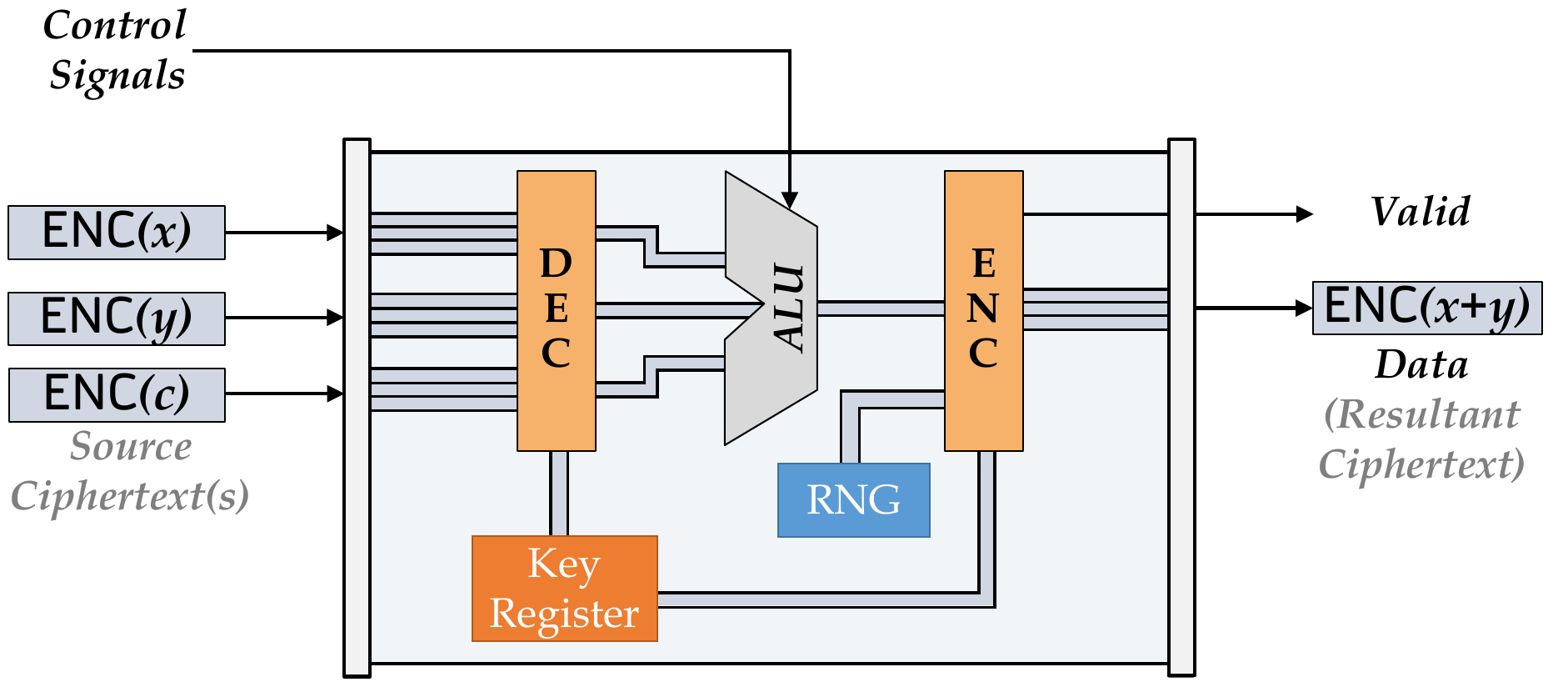}
    \caption{SE Enclave Design}
    \vspace{-2em}
    \label{fig:se_baseline}
\end{figure}

\subsection{SE Instruction Set Architecture (ISA)}
SE extends the native ISA with \stress{secure instructions} that explicitly operate on ciphertext data. 
When the processor decodes a secure instruction from the SE ISA, the instruction is dispatched to the SE Enclave for processing. 
This operation mirrors how native instructions are dispatched to functional units in a conventional processor.
Insecure instructions in the native ISA are dispatched to unsecured functional units. 
SE makes no \rev{claim} about these instructions, as they are executed in an untrusted environment. 

To ensure that sensitive data does not leave the enclave, SE restricts its secure instructions to be `data-oblivious,' only supporting arithmetic, logical, comparison, and shift operators. 
The classes of instructions supported by the SE ISA are summarized in Table~\ref{table:isa-overview}.
Specifically, the SE ISA does not support control flow or memory instructions, as these would innately leak information about sensitive data through architectural states like the program counter. Insecure versions of these instructions (\textit{i.e.}, those assumed to be operating on public data) are still present in the native ISA. 
SE enables secure control flow via an encrypted conditional move (\verb|CMOV|) instruction. 
Secure \verb|CMOV| instructions function as ternary operators, where a destination register is updated based on the value of some condition, akin to predicated instructions. This primitive enables programmers to make decisions on secret conditions, mimicking the logic of if-statements, in a safe manner. 
Finally, to perform plaintext-ciphertext operations, the processor must first encrypt the plaintext value using SE's \verb|ENC| instruction before passing the resultant value to another secure instruction for computation.


\begin{table}[!h]
\centering
\resizebox{0.95\columnwidth}{!} {
\begin{tabular}[\aboverulesep=0ex, \belowrulesep=0ex]{@{\extracolsep{0pt}}|lllcc|}
\hline
\rowcolor[HTML]{f3f3f3}\multicolumn{1}{|l}{} &  \multicolumn{2}{l}{} & \multicolumn{1}{c}{Secure}& \multicolumn{1}{c|}{Insecure}\\
\rowcolor[HTML]{f3f3f3}\multicolumn{1}{|l}{Instruction Class} &  \multicolumn{2}{l}{Example} & \multicolumn{1}{c}{(SE Enclave)}& \multicolumn{1}{c|}{(Native)}\\
&&\\[-1em]
\hline
&&\\[-1em]
Encryption & \texttt{ENC} & Encryption  & \dfull & \dnone\\
Arithmetic & \texttt{ADD} & Addition & \dfull & \dfull\\
Logical & \texttt{AND} & Logical And & \dfull & \dfull\\
Comparison & \texttt{LT} & Less Than & \dfull & \dfull\\
Shift & \texttt{SLL} & Logical Left Shift & \dfull & \dfull\\
Conditional & \texttt{CMOV} & Conditional Move & \dfull & \dfull\\
Memory & \texttt{LD} & Load & \dnone & \dfull\\
Control Flow & \texttt{JMP} & Jump & \dnone & \dfull\\
\hline
\end{tabular}
}
\caption{Summary of ISA}
    \vspace{-2em}
\label{table:isa-overview}
\end{table}

\subsection{SE Enclave Implementation}
The SE Enclave is architected as a small hardware functional unit embedded within the execute stage of the pipeline. This hardware unit includes storage of the secret key under which the user data is encrypted. When instructions are dispatched to the SE Enclave, the unit decrypts ciphertext source operands under this key, computes the requested operation in plaintext, then re-encrypts the result. 
Namely, the syntax for a secure \verb|ADD| instruction is  {\small$\mathsf{ENC(~DEC(}$}\verb|r1|{\small$\mathsf{)}$} \verb|ADD| {\small$\mathsf{DEC(}$}\verb|r2|{\small$\mathsf{)~)}$}
This operation is illustrated in Figure~\ref{fig:se_baseline}.
The resultant ciphertext value is the only value that leaves the enclave. 
The SE Enclave implementation is trusted and assumed to be free of direct data leakage.  
Further, the SE Enclave is implemented to have data-independent timing, such that timing the execution of the enclave cannot reveal the plaintext values of ciphertexts. 
For example, the SE Enclave cannot contain hardware optimizations that accelerate instructions for specific inputs (\textit{e.g.}, forwarding for multiply-by-zero~\cite{grossschadl2009side}), as these optimizations leak information via timing.

While the design in Figure~\ref{fig:se_baseline} implements the syntax of the SE ISA, the SE Enclave can have many instantiations, including different implementations of encryption and decryption. 
Below, we discuss some of the different variations of SE Enclave presented in \cite{se-paper}.
In this work, we present different RTL instantiations in Section~\ref{sec:evaluate}.

\subsubsection{Encryption Scheme}

\rev{
In SE, a value ($m$) is encrypted by appending a fresh random salt ($u$), then applying a strong pseudorandom permutation (PRP), such as a block cipher. 
In this encryption scheme, both the data value and salt are 64 bits long, thereby producing a 128-bit ciphertext. 
This scheme can be implemented in hardware with a variety of symmetric or asymmetric encryption ciphers built as rolled or unrolled implementations.}
Prior work analyses three \rev{cryptographic ciphers} for use within SE: AES-128, QARMA, and Simon. In this work, we also consider the popular asymmetric public-private key cryptosystem RSA.

The SE encryption scheme salts each ciphertext value with \rev{fresh randomness through the True Random Number Generator (TRNG)} to ensure that ciphertexts are diversified to thwart cryptanalysis attacks \rev{and some side-channel attacks, \textit{i.e.}, CIPHERLEAKS~\cite{li2021cipherleaks}. This requirement is specified in the ISA and can be ensured by a structural RTL check}. 


\subsubsection{Optimizations}
SE can also be implemented with internal optimizations that seek to improve its performance. For example, \cite{se-paper} proposes caching recently decrypted ciphertexts in order to bypass the decryptor for instructions with data dependencies. 
Optimizations have the potential to invalidate the security \rev{claims} of SE. 
In this work, we assess several secure and insecure optimizations of the SE Enclave to demonstrate that the secure optimizations are fully verified and the insecure optimizations are detectable by our novel verification technique.

\subsection{Opportunities for Hardware Verification}

SE's security \rev{claims} are based on the assumption that the hardware is secure. 
The SE architecture positions itself well for formal verification because it has a minimal hardware footprint and possesses no trusted software, which is in sharp contrast to other \rev{Trusted Execution Environments} (TEEs), like Intel SGX. As noted by prior work~\cite{costan2016intel}, formally reasoning about existing enclaves including Intel SGX is infeasible as any proof would have to model all processor features that exposed registers. Further, such work would be short-lived as any architectural modifications would invalidate this security proof.
Rather, SE's compact design allows it to be verified independently of other structures of the processor, thus avoiding the deficiency of verifying the whole processor.
In this work, we formally verify SE's security claims to establish trust in the SE ISA and its implementation. 


\section{Threat Model}\label{sec:threat-model}

In this section, we describe and discuss our threat model. We first present our security goals. This is followed by listing the attacker's capabilities which can compromise these goals. Finally, we describe the root of trust which specified the components that can be assumed to be trusted.
\subsection{Security Goals}
In this work, we will formally verify that SE ISA and SE Enclave RTL implementations preserve data confidentiality. Specifically:
\begin{itemize}[leftmargin=*]
    \item Any program should not leak sensitive data through architectural states and/or side channels through instructions from the SE ISA.
    \item Any SE secure instruction or sequence of instructions should not leak sensitive data through microarchitectural states and IO signals of the SE Enclave.
\end{itemize}

\subsection{Attacker Capabilities}
In this work, the attacker's goal is to, within a reasonable amount of time, infer the plaintext secret values by analyzing the program's computation results, snooping on SE Enclave's IO signals, or analyzing the program execution time. We consider attackers to possess the following capabilities:
\begin{itemize}[leftmargin=*]
    \item The attacker can observe and/or arbitrarily change digital signals outside of the SE Enclave including signals on SE Enclave IO. \rev{These signals can be observed at every cycle even though each cycle may not result in an architectural state update.}
    \item The attacker can run any program using secure/insecure instructions, and measure the program execution time.
    \item The attacker \textbf{cannot} measure and/or arbitrarily change the states inside the SE Enclave and the physical characteristics of the chip.
\end{itemize}

\subsection{Root of Trust}\label{subsec:root-of-trust}
This work assumes the following components and algorithms, which have been studied extensively in prior works, can be trusted and meet design requirements:
\begin{itemize}[leftmargin=*]
    \item \textit{True Random Number Generator (TRNG).} We assume there exist TRNG designs that are capable of supplying at least \rev{$s$ (for example, $s=64$ in our designs}) bits of random number per cycle. 
    Prior works have proposed TRNGs with different designs~\cite{chaoticswitch,broadbandnoise,prng,jitter-rng, 7008853, 1317067, noiseic,chaoticjitter}. 
    Recent laser-based random number generators can generate up to 250 terabits per second \cite{kim2021massively}. \rev{This approach is also shown to be bias-free in nature\cite{Wei:09}.}
    \item \textit{Key Exchange Mechanism.} We assume there exists a safe key exchange mechanism for the user and SE Enclave to establish shared keys. 
    Key exchange methods like RSA key exchange, Diffie-Hellman (DH), and Elliptic-curve Diffie-Hellman (ECDH) have been well-studied \cite{harris2006rsa, rescorla1999diffie, bresson2007provably, mehibel2017new}.
    Public key infrastructure (PKI)\cite{maurer1996modelling,blaze1998keynote, housley1999internet, weise2001public} has been developed over the years to provide authenticity guarantees. 
    \rev{In SE, key exchange can be done using a standard small-footprint hardware-only implementation in the Enclave, as in typical Hardware Security Modules~\cite{attridge2002overview}, which puts the key directly into the key register and does not interact with the rest of the Enclave. The only output of the key register is shown in Figure~\ref{fig:se_baseline} and our verification proves the key is secure {\it once in the key register}. Verifying key exchange is orthogonal to this work.}
    \item \textit{Encryption Scheme.} We assume popular encryption algorithms are strong and robust against crypto-analysis when paired with long enough keys. Multiple hardware-friendly encryption schemes have been proposed and can be used in SE, \textit{e.g.}, AES-128\cite{heron2009advanced, hamalainen2006design}, Simon-128/128\cite{7167361}, QARMA$_{11}$-128-$_{\sigma 1}$\cite{avanzi2017qarma}.
\end{itemize}
A more detailed account of RTL-level assumptions that flow from the above threat model will be discussed in \S~\ref{subsec:assumptions}.

\section{SE ISA Modeling and Analysis}\label{sec:isa}

A conventional ISA does not place any requirements on microarchitectural states and state transitions as it only requires functional correctness. However, the SE ISA needs to make explicit and verifiable rules such that the system can guard off software-based attacks.
On the hardware side, the microarchitectural implementations need to faithfully follow the design requirements set by the ISA in a verifiable manner. The security properties resulting from these rules should be attested with formal proofs on the implementation side.

From the perspective of a program defined in such an ISA, there are two types of data: user \stress{private data} and \stress{public data}. Private data are sensitive data that are always encrypted under secured keys. The value of private data should only be visible to trusted and formally verified hardware components and remain invisible to any software or untrusted components. Public data are non-sensitive data that are stored in plaintext, \textit{i.e.}, program code, public constants, etc. While the conventional ISA is sufficient to operate on public data, a specialized set of instructions needs to be implemented in order to process private data securely. This set of secure instructions specifies additional requirements, which are not needed for public data only programs, to securely execute programs that compute on secured sensitive data. While previous work \cite{se-paper} has described these properties informally, we give precise definitions of the properties in this work and state proof obligations for the RTL verification to ensure that the implementation is faithful to the ISA specification. These requirements can be organized into two groups: \stress{Direct Disclosure Safety} and \stress{Indirect Disclosure Safety}. We will now consider them one at a time in the following sections.

\subsection{ISA Direct Disclosure Safety Requirements}
\label{sec:disclosure}

SE ISA must define the trust boundary between secure instructions and insecure instructions. This boundary is then guaranteed by the architecture. Inappropriate disclosure at the ISA level can happen in two ways: first, instructions could directly disclose the private plaintext data to insecure locations (\textit{e.g.}, registers that are readable by the attacker); second, the ciphertext is produced by a weak encryption cipher, making it vulnerable to cryptanalysis attacks. 
Therefore, for any valid SE program, which is a sequence of instructions defined by the secure ISA, we have the following two requirements. First, we put the obvious restriction that there are no direct disclosures; and second, we put a requirement on the quality of the ciphertexts produced by the secure instructions.

\subsubsection{Direct Information Disclosures}%
\label{subsub:direct-disclosure}
SE instructions take either a ciphertext or a public plaintext and produce a ciphertext as a result. Any secret data must never be disclosed to insecure (\textit{i.e.}, non-SE) instructions. This is reflected at both architectural and microarchitectural levels since SE uses encryption to hide sensitive information from untrusted software and hardware. The encryption module marks the trust boundary between the hardware components since all data is re-encrypted before it is emitted out of the trusted execution component and exposed to untrusted software and hardware. Encrypted data can be safely stored in untrusted storage or computed on by an insecure instruction. 

The instruction level properties lift to programs naturally. All SE instructions produce ciphertext and there is no decrypt instruction in the ISA. Thus, non-SE instructions cannot get access to the plaintext. From this, we conclude that any composition of SE instructions and non-SE instructions cannot disclose the secret data. 
\rev{
\subsubsection{Quality of Ciphertexts}
\label{subsub:quality-of-crypto}

Ciphertexts need to be safe from cryptanalysis when disclosed to an attacker of reasonable strength.
In this work, we consider security against chosen-ciphertext attacks (CCA), a standard attack model for cryptanalysis where the attacker has oracle access to the encryption function and an arbitrary collection of old ciphertext-plaintext pairs from the SE Enclave.
This model is more powerful than that of chosen plaintext attacks (CPA). 
We assume that the attacker has bounded computational power and storage space, and does not have access to the encryption key a priori. 
If the encryption scheme used by the SE Enclave has indistinguishability under CCA (termed CCA Security), then an attacker has close to a random-guess success rate in learning the plaintext value of a new ciphertext released by the SE Enclave, even if the attacker has amassed an arbitrary collection of previous ciphertext-plaintext pairs. In the following sections, we detail the encryption scheme used by the SE Enclave and how it achieves CCA security. 

The single-instruction CCA security described above immediately lifts to programs involving multiple instructions since CCA security in the single-message case implies CCA security in the multiple-messages scenario. 

}

\subsection{ISA Indirect Disclosure Safety Requirements}
\label{sec:side-channel}
Another important requirement that the ISA specifies and the underlying design should enforce is that there cannot be any observable digital effects correlated to the plaintext value of the secrets. The most prominent indirect disclosure is based on the digital side channel of the program's timing behavior. Similarly, the control flow and memory access pattern of the program can produce information about private user data without direct disclosure. The ISA must restrict all three indirect disclosures.

\subsubsection{Instruction Stream Side-Channels}%
\label{subsub:instruction-stream}

It is necessary that knowledge of which instructions get executed does not yield any more information than what is known at compile time. 
Intuitively, such information can be derived only if there is a relationship between the control-flow structures of SE and the secret data of the user. However, insecure control-flow instructions do not see secret values by the definition of the ISA thus control-flow decisions cannot give any information about the secret. The only control-flow instruction that operates on secret data is the secure \verb|CMOV| instruction. However, as a secure SE instruction, it always produces ciphertexts with the \rev{property described in \S\ref{subsub:quality-of-crypto}}. 
Unlike general control-flow instructions, like \verb|jmp|, which influence which instructions are executed, \verb|CMOV| always moves encrypted (and thus indistinguishable\rev{)} data into the same destination register. 
Therefore, although \verb|CMOV| makes a move decision based on secret data, it does not expose the secret. Consequently, observing the stream of executed instructions does not give the attacker a better-than-chance shot at guessing the secret. This property lifts to the program level from the instruction's property since all other disclosures from secure instructions to insecure instructions are removed by the direct disclosure restrictions from \S~\ref{subsub:direct-disclosure}. 

\subsubsection{Address Stream Side-Channels}%
\label{subsub:address-stream}
Side channels in programs can also exist in their address streams where memory access patterns can expose information about private data indirectly. Similar to \S~\ref{subsub:instruction-stream} above, such indirect disclosure is possible only when there is a relationship between the user's private data and memory access instructions. However, the insecure memory instructions, such as \verb|load| and \verb|store|, do not have access to secret values. The only instruction that writes data to a location is the secure \verb|CMOV| instruction, which we have argued in \S~\ref{subsub:instruction-stream} to be side-channel free. 

Since all secure SE instructions do not disclose secrets by \S~\ref{subsub:direct-disclosure} and the ciphertexts produced \rev{have the property described in \S\ref{subsub:quality-of-crypto}}, we conclude that all programs composed of SE instructions do not leak secret information through their address access streams. 

\subsubsection{Instruction Timing Side-Channels}%
\label{subsub:instruction-timing}
A final possible side channel we consider is timing variance in secure instructions that can leak information about the secret data indirectly. First, we require all secure instructions not to have timing variance that correlates to secret data. By the direct disclosure argument in \S\ref{subsub:direct-disclosure} we know that insecure SE instructions do not have access to secret data thus any timing variance they might have in their execution (including nondeterministic variance) cannot be correlated to secret values. From these two properties of instructions, we can lift to program-level properties naturally. Specifically, a program composed of SE instructions executed on two different secret data will have exactly the same timing behavior if we fix the rest of the execution environment. All variations that might occur in this setting would be caused by insecure instructions possibly exhibiting nondeterministic behavior, which is always independent of the secret data. Thus the timing of programs does not reveal any secrets.

\subsection{Formalization of ISA Properties}
\label{sub:formalization-isa}
While designing an ISA is already a challenging task that must carefully tread the interface between the hardware's behavior and the software's expectations, stating security principles at the ISA level is an even harder task. Fortunately, the SE ISA comprises a few unique instructions with precise semantics and security requirements. This allows us to give a formal treatment to some of the informal descriptions presented in \S~\ref{sec:disclosure} and \S~\ref{sec:side-channel} regarding disclosure and side channels, respectively. To formally reason about the ISA level properties we first start by defining a minimal syntax recognized by the secure SE processor. We allow an arbitrary but finite set of registers observable by the attacker except for the \verb|keyReg| which is private to the secure SE Enclave. The instruction $\mathsf{enc~r_1}$ is the encryption instruction that encrypts the value in the parameter register and places the result back into the same register. We abstract all binary operations, \textit{e.g.}, addition, subtraction, and shifts, to the instruction $\mathsf{bop~r_1, r_2}$ where some operation is performed on the values of the parameter registers, and the result is placed into the first parameter. Finally, a \verb|cmov| is written as an if-else statement with a single assignment instruction as the body. Note that, unlike standard imperative languages, we do not allow conditional \verb|if| statements with arbitrary bodies and there is no \verb|while| construct either. The sequential composition operator $\mathsf{c;p}$ is a syntax restriction, forcing it to be right-associative. Enforcing it syntactically simplifies the semantics and type system, yet has no significant effect since sequential composition is associative. Figure \ref{fig:syntax} presents this syntax of the SE language.

\begin{figure}[h]
    \centering
    \begin{bnf*}
        \bnfprod{Registers $e$}{\bnfts{$r_1$} \bnfor \cdots \bnfor \bnfts{$r_{n}$} \bnfor {\sf keyReg} \bnfor b \bnfor [b]}\\
        \bnfprod{Commands $c$}{\bnfts{enc $r_1$}  \bnfor \bnfts{bop $r_1,r_2$} \bnfor \bnfts{skip}}\\
        \bnfmore{\bnfts{if~$r_1: r_2 \leftarrow r_3~\mathsf{else}~r_2 \leftarrow r_4$}}\\
        \bnfprod{Program $p, q$}{{c} \bnfor {c; p}}
    \end{bnf*}
    \caption{Minimal syntax for the SE ISA.}
    \label{fig:syntax}
\end{figure}

The semantics of the language is given by the small-step semantics in Figure \ref{fig:small-step}. The semantics is entirely standard in imperative languages with the exception of the encryption and decryption operations and the time function $t$. First, all values are defined on finite length bits, \textit{i.e.}, $bits = \{0, 1\}^\rev{j}$ \rev{where $j \in \mathbb{N}$ is a non-deterministically picked natural number}, which allows bit-level operators to be well defined in \verb|bop|. Bits of all zeros can be interpreted as a \rev{boolean} false and, conversely, bits of all ones can be interpreted as a \rev{boolean} true. 

\rev{
Next, we use a simple and well-known encryption scheme which we present here for completeness. Details of this construction can be found in introductory textbooks such as \cite{katz2020introduction, joyofcryptography}. Let $Func$ be the set of all functions of type $\{0,1\}^n \rightarrow \{0, 1\}^n$ and $F: \{0, 1\}^n \times \{0, 1\}^n \rightarrow \{0, 1\}^n$ be a keyed pseudorandom function inducing a distribution on $Func$. As usual, a bijective pseudorandom function $P$ is a pseudorandom permutation. We will write $F_k$ and $P_k$ defined as $F_k(\cdot) = F(k, \cdot)$ and $P_k(\cdot) = P(k, \cdot)$. A strong pseudorandom function $F$ is a pseudorandom function such that any polynomially bounded adversary $\mathcal{A}$ has a negligible advantage in differentiating between $F$ and the random function $F'$ even when given the inverse function $F^{-1}$. A strong pseudorandom permutation $P$ is a bijective strong pseudorandom function.

Now we can define the encryption scheme used by the SE enclave. Let $\mathcal{SE} = \langle encrypt(\cdot, \cdot), decrypt(\cdot, \cdot) \rangle$. The encryption key $k$ is uniformly picked from keys of length $s + n$. The message space is $\mathcal{M} = \{0, 1\}^n$, \textit{i.e.}, messages of length $n$. The encryption scheme is constructed as follows for any $m \in \mathcal{M}$:
\begin{align*}
    & u \leftarrow uniform(s) && encrypt(k, m) = P_k(m || u) \\ 
    &decrypt(k, m) = P_k^{-1}(m)[0:n] && decrypt(k, encrypt(k, m)) = m \end{align*}
\noindent
for any strong keyed pseudo-permutation $P$ of block length $s + n$, where $||$ is defined as string concatenation. The random value $u$ is sampled from the uniform distribution of strings of length $s$.

This scheme, $\mathcal{SE}$, is known to be secure against chosen ciphertext attacks (CCA-secure) because an attacker has only negligible advantage in learning about $m$ given the ciphertext $encrypt(k, m)$. In Appendix \ref{appx:crypto-def}, we give a formal definition of CCA security (Definition \ref{def:cca-security}) and show that $\mathcal{SE}$ is  CCA-secure (Theorem \ref{thm:se-security}).  

To simplify the semantics of the ISA, we implicitly make the assumption that the polynomially-bounded attacker $\mathcal{A}$ has a stricter bound on the number of queries it can make to the encryption oracle. Intuitively, this assumption guarantees that the probability of the event that the same random value $u$ is used for two different encryption oracle calls is effectively zero. Given this assumption, we take the distribution of the ciphertexts to be truly uniform.\footnote{\rev{This effectively means the encryption scheme is a perfect encryption scheme. Strictly speaking, the pseudorandom permutation of the encryption scheme needs to be swapped for a random permutation to achieve this. But, any polynomial attacker only loses a negligible advantage due to this swap.}}
This assumption only allows us to avoid the more complicated probabilistic semantics that would be required to properly model this negligible chance and otherwise has no effect on the semantics. 

In the ISA semantics, we assume there is a truly uniform function that can generate $s$ bits for every encryption query and that a strong pseudorandom permutation is used. We justify these assumptions in \S\ref{sub:isa-to-rtl} by discharging verification responsibilities to the RTL-level verification, as summarized in Table~\ref{table:isa-property}. This leaves us with two tasks at the ISA level: \textit{i)} ensure that $u$ is of some pre-fixed length $s$ on every encryption, and \textit{ii)} ensure that a fresh $u$ from the random number generator is used for each encryption query. We formalize these two tasks in the semantics in the rest of this section.
}

The state of the system is represented by a set of registers $R$ containing values of $b \in bits$ and the syntactically decorated $[b]$. Marking values representing ciphertexts with brackets allows us to consider all ciphertexts to be effectively equivalent once we define the equivalence class. The function $\sigma:R \rightarrow bits$ provides the mapping. Let $\Sigma$ be the set of state maps $\sigma$ and allow $\sigma_1, \cdots, \sigma_n$ to range over $\Sigma$. There are three syntactic categories: programs, commands, and registers. The small-step arrow $\rightarrow_r$ produces bits by reading registers. We combine programs and commands into one syntactic category for the semantics, since all commands are programs; this makes the semantics more readable. The small-step arrow $\rightarrow$ for this combined syntactic category is defined over configurations of $\langle p, \sigma \rangle$. \rev{We write $u \sim uniform(s)$ to say the value $u$ is sampled from the uniform distribution of strings of length $s$.} The function $t: \Sigma \rightarrow \mathbb{N}$ models the timing behavior of the system by allowing an arbitrary finite function to decide how long it takes for the instruction to complete execution. In a concrete design, the time taken can depend on a number of system and instruction level properties thus we allow the system to make this decision based on the entire system state $\sigma$. 

The SE language semantics is a restricted fragment of standard imperative languages. For example, the conditional move commands (\verb|CMOV-T| and \verb|CMOV-F|) reflect this in that the bodies are single assignment instructions instead of the traditional recursive bodies. Additionally, the assignment is forced to the same location in both the true and false branch of the conditional. There is no use of exposed \verb|store| or assignment operator in the language, but one could mimic it using the existing operators with the restriction that only ciphertext is written to locations. The occurrence of register symbols is treated as usual as a free variable (\verb|REG|). The encryption instruction (\verb|ENC|) encrypts the value found in the parameter register and places it back in the parameter. Similarly, the binary operations instruction (\verb|BOP|) first decrypts the values, and computes on the plaintext values using the semantic operator $\oplus \in \{+, <<, >>, -, \cdots\}$, then finally re-encrypts the resulting value using the $encrypt$ relation, and writes the ciphertext back into the first operand register. Finally, the sequence operator (\verb|SEQ|) takes single steps transforming the head of the sequence until the base case of \verb|skip| is reached. 

Since the \verb|SEQ| is the only inductive rule in the semantics, the program structure is that of a list instead of a tree as usual. Consequently, the small-step semantics never diverges in its execution. This property is essential for the soundness proof in Theorem \ref{thm:soundness}. The security reasoning is done within the type system thus keeping the semantics standard.

\begin{figure}[!h]
\vspace{-0.5em}
\begin{mathpar}
    \inferrule[cmov-t]{
        \snr{r_1}{\sigma}{[c_1]}\\
        \snr{r_3}{\sigma}{[c_3]}\\\\
        \snr{keyReg}{\sigma}{k}\\
        decrypt(c_1, k) = true\\
        decrypt(c_3, k) = m\\
        \rev{u \sim uniform(s)}\\
        \rev{encrypt(m||u, k) = [c_5]}\\
    }{\sn{\mathsf{if}~r_1: r_2 \leftarrow r_3~\mathsf{else}~r_2 \leftarrow r_4}{~\sigma}{skip}{\sigma[[c_5]/r_2]}{2}} \and
\inferrule[cmov-f]{
        \snr{r_1}{\sigma}{[c_1]}\\
        \snr{r_4}{\sigma}{[c_4]}\\\\
        \snr{keyReg}{\sigma}{k}\\
        decrypt(c_1, k) = false\\
        decrypt(c_4, k) = m\\
        \rev{u \sim uniform(s)}\\
        \rev{encrypt(m||u, k) = [c_5]}\\
    }{\sn{\mathsf{if}~r_1: r_2 \leftarrow r_3~\mathsf{else}~r_2 \leftarrow r_4}{~\sigma}{skip}{\sigma[[c_5]/r_2]}{2}} \and  
    \inferrule[reg]{
        \sigma(r_1) = b
    }{\snr{r_1}{\sigma}{b}} \and
        \inferrule[enc]{
        \snr{r_1}{\sigma}{n}\\
        \snr{keyReg}{\sigma}{k}\\
        \rev{u \sim uniform(s)}\\
        \rev{encrypt(n||u,k) = [c_1]}
    }
    {\sn{\mathsf{enc}~r_1}{\sigma}{skip}{\sigma[[c_1]/r_1]}{1}}\and
    \inferrule[bop]{
        \snr{r_1}{\sigma}{[c_1]}\\
        \snr{r_2}{\sigma}{[c_2]}\\
        \snr{keyReg}{\sigma}{k}\\
        decrypt(c_1, k) = n\\
        decrypt(c_2, k) = m\\
        \rev{u \sim uniform(s)}\\
        \rev{encrypt((n\oplus m) || u, k) = [c_3]}
    }{\sn{\mathsf{bop}~r_1~r_2}{\sigma}{skip}{\sigma[[c_3]/r_1]}{2}}\and
    \inferrule[seq]{\ }{\sn{\mathsf{skip; q}}{\sigma}{q}{\sigma}{1}}\and
    \inferrule[]{\sn{\mathsf{c}}{\sigma}{\mathsf{skip}}{\sigma'}{1}}{\sn{\mathsf{c;p}}{\sigma}{skip;p}{\sigma'}{1}}
\end{mathpar}
\vspace{-2mm}
\caption{The small-step semantics of the SE Enclave.}
\vspace{-5mm}
\label{fig:small-step}
\end{figure}


The type system used to reason about the security of information flow in the SE language is presented in Figure \ref{fig:type-system}. The types are generated by the following grammar
\vspace{-1mm}
\begin{bnf*}
        \bnfprod{Security labels $\ell$}{\bnfts{\pu} \bnfor \bnfts{\pr}}\\
        \bnfprod{Program Types $\tau$}{\bnfpn{$\ell$}~ prog \bnfor \bnfpn{$\ell$}}\\
\end{bnf*}
\vspace{-8mm}

Let $L = \{\ell_1, \cdots, \ell_n\}$ be the set of security labels and $\mathcal{L} = \langle L, \leq \rangle$ by the bounded security lattice generated by $L$. We assume $L = \{\pu, \pr\}$ and that $\pr$ is the top of the lattice and $\pu$ is the bottom of the lattice. Generally, any security lattice can be decomposed into a low and high partition so our assumption is without loss of generality. 
The typing environment $\G$ maps register locations to their types. We consider all register locations to be $\mathsf{\pu}$ except the \verb|keyReg| which is marked as \rev{private}, thus,
\[
\G(r) = \begin{cases}
\mathsf{\pr}, &r = keyReg\\
\mathsf{\pu}, &\text{otherwise}
\end{cases}
\]

The rule \verb|REG| simply states this in the inductive rules. All constant values from $bits$ start out as $\pum$ (rule \verb|CONST|). Similarly, the \verb|skip| instruction is always typed as $\pum$ (rule \verb|SKIP|). The instruction \verb|seq| is typed as $\pum$ if the two programs that are composed are already typed as $\mathsf{\pud}$. The instructions \verb|CMOV| and \verb|BOP| are only required to demonstrate their operands are $\pum$ and can be immediately typed as $\mathsf{\pud}$.\footnote{\rev{It is not always the case that functions with $\mathsf{\pu}$ inputs will output $\mathsf{\pu}$ outputs. These typing rules for CMOV and BOP are sound because the security features of the semantics ensures the output is encrypted with fresh salt.} }

\begin{figure}[!htb]
\begin{mathpar}
  \inferrule[reg]{\G(r) = \ell}{\tp r \ell}\and
  \inferrule[const]{ b \in bits}{\tp b \pu}\and
  \inferrule[enc]{
    \tp {r_1} \pu
  }{\tp {enc ~r} \pud }\and
  \inferrule[skip]{\ }{\tp {skip} \pud }\and
  \inferrule[seq]{
    \tp {p_1} \ell'~prog \\
     \tp {p_2} \ell''~prog\\
     \ell = \ell' \sqcup \ell''
  }{\tp {p_1; p_2} \ell~prog}\and
  \inferrule[bop]{
    \tp{r_1} \pu \\
    \tp{r_2} \pu \\
  }{\tp {bop~ r_1~r_2} {\pud} }\and
  \inferrule[cmov]{
    \tp {r_1} \pu \\
    \tp {r_2} \pu \\
    \tp {r_3} \pu \\
    \tp {r_4} \pu 
  }{\tp  {if~r_1:r_2 \leftarrow r_3~else~r_2 \leftarrow r_4} \pud} \and
\end{mathpar}
\vspace{-2mm}
\caption{A high-level security type system for the secure SE language.}
\vspace{-2mm}
\label{fig:type-system}
\end{figure}

Next, we define \rev{an equivalency of states of the system and, specifically,} low-equivalency following the convention in the literature \cite{volpano1996sound, ferraiuolo_hyperflow_2018, myers2011proving}. 

\begin{definition}[Low-equivalent]
\label{def:equiv}
First, on the finite length bits $b$ and $[b]$, we define the equivalence class generated by the rules\footnote{\rev{This equivalence class is justified by the fact that the ciphertexts have a uniform distribution. Thus to the attacker, any two ciphertext carry the same information and the attacker shouldn't have a reasonable preference between any two ciphertexts.}},
\begin{mathpar}
    \inferrule[equiv-br]{b, b' \in bits }{[b] \approx [b']}\and
    \inferrule[equiv]{b, b' \in bits \\ b = b'}{b \approx b'}
\end{mathpar}
Now we can define low-equivalence. In context $\Gamma$, states $\sigma, \sigma'$ are low equivalent $\sigma \approx_l \sigma'$ if they are equivalent on all low locations,
 \begin{align*}
   \Gamma\vdash\sigma(r) \approx \sigma'(r) \text{ for all $r$ where } \G(r) \leq l
 \end{align*}
\end{definition}
We can now restrict the behavior of $t$ using the definition above. We require that $t$ decides the amount of time taken by the instruction only using \verb|public| data,
\begin{align}
\label{eq:time-invariance}
\text{ if } \Gamma\vdash\sigma \approx_l \sigma'  \text{  then }t(\sigma) = t(\sigma')
\end{align}
Intuitively, this means the timing function $t$ is influenced only by the location of level $l$ or lower (\textit{i.e.}, there is no timing dependency between the private data and the public data that could lead to timing side channels). \textit{This corresponds to the instruction timing property in \S~\ref{subsub:instruction-timing}.}

We just need one preliminary security lemma toward the main theorem now. For a well-typed program $c$, taking one step in the small-step semantics on two different states of the system that agree on publicly visible state locations will always produce states that continue to agree on publicly visible state locations. Moreover, the timing behavior of $t$ will also be equivalent on the two final states. This means changing any private locations in the state will not have an observable change on the publicly observable locations of the output states or the timing behavior of the program. The attacker that can observe only public locations (thus not inside of SE) cannot tell the difference between two executions of a program where the private data might be different and thus cannot derive additional knowledge about the secret user data from the public data. This property corresponds to \S\ref{subsub:direct-disclosure}. Lemma \ref{lma:single-command} states this property formally now. 

\begin{lemma}[Single Step Security]
\label{lma:single-command}
If
\begin{enumerate}
    \item $\tp c {\ell}$
    \item $\G \vdash \sigma_1 \approx_l \sigma_2$
    \item $\langle c, \sigma_1 \rangle \xrightarrow{t(\sigma_1)} \langle \mathsf{c_1'}, \sigma_1' \rangle$ 
    \item $\langle c, \sigma_2 \rangle \xrightarrow{t(\sigma_2)} \langle \mathsf{c_2'}, \sigma_2' \rangle$ 
    \item $dom(\G) = dom(\sigma_1) = dom(\sigma_2)$
\end{enumerate}
then we have $\G \vdash \sigma_1' \approx_l \sigma_2'$ and $t(\sigma_1) = t(\sigma_2)$
\end{lemma}
\begin{proof}[Proof Sketch]
By induction on structure of the derivation of $\langle c, \sigma_1 \rangle \xrightarrow{t(\sigma)} \langle \mathsf{c_1'}, \sigma_1' \rangle$. The timing requirement is immediate from assumption (2) and Eq. \ref{eq:time-invariance}. See Appendix \ref{appx:proofs} for the full proof.  
\end{proof}

The security result can now be stated via the soundness of the type system. The soundness argument of Theorem \ref{thm:soundness} follows from Lemma \ref{lma:single-command} for the most part by generalizing the number of steps to an arbitrary number (\textit{e.g.}, multiple steps until the program is equivalent to \verb|skip|).

\begin{theorem}[Soundness]
\label{thm:soundness}
If
\begin{enumerate}
    \item $\tp c {\ell}$
    \item $\G \vdash \sigma_1 \approx_l \sigma_2$
    \item $\langle c, \sigma_1 \rangle \xrightarrow{n}_* \langle \mathsf{skip}, \sigma_1' \rangle$ 
    \item $\langle c, \sigma_2 \rangle \xrightarrow{m}_* \langle \mathsf{skip}, \sigma_2' \rangle$ 
    \item $dom(\G) = dom(\sigma_1) = dom(\sigma_2)$
\end{enumerate}
then we have $\G \vdash \sigma_1' \approx_l \sigma_2'$ and  $n = m$.
\end{theorem}
\begin{proof}[Proof Sketch]
By induction on the number of steps. Both the base case and inductive case are consequences of Lemma \ref{lma:single-command}. The timing property follows from the induction as well. See Appendix \ref{appx:proofs} for details.
\end{proof}
 
The type system and its soundness formalize the properties discussed in \S\ref{sec:isa} except \S\ref{subsub:instruction-stream} (Instruction stream) and \S\ref{subsub:address-stream} (Address Stream).  \S\ref{subsub:instruction-stream} and \S\ref{subsub:address-stream} refer to insecure instructions that can access all $r$ where $\tp r \pu$ but cannot decrypt ciphertexts since \verb|keyReg| is \verb|private|. An insecure instruction can at most corrupt ciphertexts but when executed on two low-equivalent states must produce a pair of low-equivalent states as well. Any timing variation in these instructions is not correlated to private data since the behavior of $t$ is restricted by low-equivalence. In the following section, we summarize the requirements the ISA places on the RTL-level to provide the security properties of Theorem \ref{thm:soundness}.

\subsection{Summary of Proof Obligations Discharged to RTL Verification}
\label{sub:rlt-obligations}

To achieve the above program-level properties, we expect the RTL implementation to satisfy some requirements about individual secure SE instructions. Note that we have no requirements for insecure instructions. We enumerate these requirements here and later show how the RTL-level verification formally guarantees these instruction-level properties. The properties are split into instruction-level properties which are properties that must be satisfied by each instruction, and a system-level property which is a more general requirement on the system. 

\rev{
In Table~\ref{table:isa-property}, property (P1) summarizes \S\ref{subsub:quality-of-crypto} and is required by all the direct and indirect disclosure arguments in \ref{sec:disclosure} and \ref{sec:side-channel} respectively. Both (P1.1) and (P1.2) are formalized in \S\ref{sub:formalization-isa} and Appendix \ref{appx:crypto-def}.} Property (P2) is required by \ref{subsub:instruction-timing} and Eq. \ref{eq:time-invariance} formally states this requirement. 
Finally, the system-level property (P3) is required by all the disclosure arguments since each assumes there is no way to get a plaintext from a ciphertext \rev{using the SE Enclave's key} except within the small SE Enclave. The formal semantics in Figure \ref{fig:small-step} and the concluding remarks regarding insecure commands in \S\ref{sub:formalization-isa} require it as well.

\begin{table}[ht]
\centering
\resizebox{\columnwidth}{!} {
\begin{tabular}[\aboverulesep=0ex, \belowrulesep=0ex]{@{\extracolsep{0pt}}|cc|c|c|}
\hline
\rowcolor[HTML]{f3f3f3} & \multicolumn{3}{c}{Instruction-Level Property}\vline \\
\hline
\multicolumn{1}{|c|}{\multirow{2}{*}{P1}} & Strong &P1.1 & \rev{$encrypt$ is a strong keyed PRP; $decrypt$ is its inverse}\\
\cline{3-4} 
\multicolumn{1}{|c|}{}& ciphertext  &P1.2& \rev{The value $u$ has length $s$ and is sampled from a uniform dist.}\\
\hline
\multicolumn{1}{|c|}{P2} & \multicolumn{3}{c}{Instructions have no secret dependent time variation}\vline \\
\hline
\rowcolor[HTML]{f3f3f3} & \multicolumn{3}{c}{System-Level Property}\vline \\
\hline
\multicolumn{1}{|c|}{P3} & \multicolumn{3}{c}{Re-encrypt all data for output. No decrypt operation outside SE Enclave} \vline \\ 
\hline
\end{tabular}
}
\caption{Properties assumed at the ISA level and discharged to the RTL verification.}
\vspace{-7mm}
\label{table:isa-property}
\end{table} 

\section{RTL Security Properties}
\label{sec:definition}

In this section, we formally define the security properties in an RTL implementation of SE so as to support the instruction-level requirements from \S~\ref{sec:isa}.
\rev{We employ standard \emph{hardware information flow tracking (IFT)} to check RTL-level properties.
Briefly, if secret variables (plaintext and crypto key) do not leak to outputs of the Enclave, then the attacker is not able to infer secret values based on any observation and thus security is guaranteed.}

For the rest of this section, we will first discuss the \rev{security goal} at the RTL level, clarify the connection between ISA-level assumptions and RTL-level properties, then provide a formal definition for the \emph{hardware information flow} properties to be checked for the RTL design.

\subsection{\rev{RTL-Level Security Goal}}\label{subsec:assumptions}
\rev{
We assume that the attacker can observe and control any signal, register, and memory location outside the boundary of SE Enclave (including the Enclave's outputs), but they cannot observe and manipulate signals and states within SE Enclave. 
SE Enclave is connected with the rest of the system with a well-defined IO interface. 
Such an assumption keeps the footprint of trusted RTL design minimal. 

Our RTL-level security goal is to prevent secret variables inside the Enclave from leaking to the outputs of the Enclave in two forms: functional leakage and timing leakage.
Functional leakage happens if the attacker can directly infer secret information from the result of the SE Instruction.
Timing leakage happens if the execution time of some instruction depends on the secret and the attacker may infer secret information by measuring the execution time.

}

\begin{table}[ht]
\vspace{-1mm}
\centering
\resizebox{\columnwidth}{!} {
\begin{tabular}[\aboverulesep=0ex, \belowrulesep=0ex]{@{\extracolsep{0pt}}|c|c|}
\hline
\rowcolor[HTML]{f3f3f3} \rev{ISA-Level Requirement} & \rev{RTL-Level Property}\\
\hline
\rev{P1}& \rev{Functional Correctness of Crypto and RNG} \\
\hline
\rev{P2} & \rev{No Timing Leakage at the SE Enclave Outputs} \\
\hline
\rev{P3} & \rev{No Functional Leakage at the SE Enclave Outputs}\\
\hline
\end{tabular}
}
\caption{\rev{Mapping between ISA-Level Requirement and RTL-Level Properties.}}
\vspace{-6mm}
\label{table:isa-rtl-mapping}
\end{table}

\rev{
\subsection{ISA-RTL Property Mapping}
\label{sub:isa-to-rtl}
Table~\ref{table:isa-rtl-mapping} provides the connection between requirements from the ISA-level and the properties to be checked at the RTL-level.

\vspace{1mm}\noindent \textit{P1:} The implementation is required to satisfy the requirements of the scheme $\mathcal{SE}$ defined in \S\ref{sub:formalization-isa} without the restriction of the semantics on the number of calls the attacker makes. Thus, the implementation supports the strictly stronger attacker of the CCA-security game in Appendix \ref{appx:crypto-def}. Two specific requirements are discharged on the cryptographic algorithms: the use of a strong pseudorandom permutation and availability of a truly random $s$-bit generator. For example, the RTL may implement AES-128 which is a block-cipher (thus $s=64, n=64$), and block-ciphers are an implementation of strong pseudorandom permutation~\cite[Ch. 3.6.4]{katz2020introduction}. Other implementation options, such as RSA, come from the closely related family of trapdoor permutations and satify this requirement as well.  
The values $u$, which is defined in \S\ref{sub:formalization-isa}, is generated by a hardware-based TRNG as shown in Figure~\ref{fig:se_baseline}. As discussed in \S\ref{subsec:root-of-trust}, we assume the existence of high-bandwidth and bias-free TRNGs and consider their design out of scope for this work. The quality of the random number generator can be checked using existing tools, \textit{e.g.}, Dieharder~\cite{brown2018dieharder}. Similarly, the functional correctness of the encryption and decryption units can be checked by existing techniques~\cite{kiniry2018formally} and is not discussed in this paper.


\vspace{1mm}\noindent \textit{P2:} Any secret-dependent execution time will result in timing leakage, \textit{i.e.}, secret dependent variation in the timing at which results are available at the SE Enclave outputs. Thus, we need to check there is no timing leakage at the SE Enclave outputs in the RTL.

\vspace{1mm}\noindent \textit{P3:} Any unencrypted secret at the data output would result in a functional leakage of the instruction result at the SE Enclave outputs. Thus, we need check there is no functional leakage at the SE Enclave outputs in the RTL.
Further, without any information (functional or timing) about the encryption key (which is only in the SE Enclave), it is impossible to decrypt the ciphertext outside of the Enclave. 

The notion of `leakage' can be formalized using standard hardware IFT, which we will define for the SE Enclave and discuss in the following subsections. 
}


\subsection{SE Definition}
We use the following finite state machine (FSM) to represent an SE Enclave: $SE = (I,O,S,S_0,N,F)$, where:
\begin{itemize}[leftmargin=*]
    \item $I$ is a vector of input variables and the domain of $I $ is $ \mathbb{I}$.
    \item $O$ is a vector of output variables. The domain of $O$ is $\mathbb{O}$.
    \item $S$ is a vector of state variables and the domain of $S $ is $ \mathbb{S}$.
    \item $S_0$ is the initial value for $S$.
    \item $N: (\mathbb{S}\times \mathbb{I})\rightarrow \mathbb{S}$ is the next state function for $S$.
    \item $F: (\mathbb{S}\times \mathbb{I})\rightarrow \mathbb{O}$ is the output function.
\end{itemize}

The secrets, \textit{i.e.}, plaintext output $p$ of decryption and the secret key $k$ are also state variables in $S$.
$S$ also includes the random salt $r$ used as input to the encryption unit.

To describe the execution of the SE Enclave, we introduce the notion of a finite-length execution trace: 
$\Pi = (\pi_I, \pi_O, \pi_S) $ represents an execution for $n$ cycles where $\pi_I=(I_0, I_1, I_2, ... I_{n-1})$ is a trace of input values for each cycle, $\pi_O=(O_0, O_1, O_2, ... O_{n-1})$ is a trace of output values for each cycle, $\pi_S=(S_0, S_1, s_2, ... S_{n-1})$ is a trace of state variable values from each cycle.
The elements in $\pi_I, \pi_O, \pi_S$ follow the next state function $N$ and the output function $F$.
In the following parts of the paper, we may use $\pi_x=(x_0, x_1, x_2, ... x_{n-1})$ to denote the trace of a variable $x$ which consists of the value of $x$ at every cycle in the execution, and $\pi_Y=(Y_0, Y_1, Y_2, ... Y_{n-1})$ to denote the trace of a vector $Y$ which consists of the value of $Y$ at every cycle in the execution.

Next, we define information flow using execution traces.

\subsection{Classic Information Flow Definition}

The classic information flow definition is based on the well-known `non-interference' property~\cite{clarkson2010hyperproperties,goguen1982security}.
The key idea of this property is: if a variable $x$ never interferes with another variable $y$ in the system, then replacing $x$ with different values will never affect the value of $y$.
Let $s$ be a secret variable and $o$ be an output variable. In our setting, if there is no information flow from $s$ to $o$, then the value of $o$ is independent of the value of $s$, thus it is impossible to infer $s$ based on $o$.

$s$ not influencing $o$ means the value of $o$ at every cycle is not changed when the value of $s$ is replaced with an arbitrary value. 
Let $\Pi$ be an execution trace of length $n$, $Q$ be the vector state variables besides $s$, $\pi_s$ denote the trace of $s$, $\pi_{Q}$ denote the trace of $Q$, $\pi_o$ denote the trace of $o$, $F_o$ denote the output function for $o$, $N_{Q}$ denote the next state function for $Q$.
If in $\Pi$ we replace $\pi_s$ with a different trace $\pi_s'$ (it differs from $\pi_s$ in at least one cycle), then we can use $\pi_s'$ to compute the new trace of $Q, o$, \textit{i.e.}, $\pi_Q'$ and $\pi_o'$ as follows:
$$ Q_0' = Q_0, Q_i' = N_Q(I_{i-1}, Q_{i-1}', s'_{i-1}), 1\leq i<n$$ $$o_i'=F_o(I_i, Q_i', s'_i),0\leq i < n$$
There exists information flow from $s$ to $o$ if and only if
$$\exists \Pi, \exists \pi_s'$$
such that after computing $\pi_Q'$, $\pi_o'$,
$$\pi_o\neq \pi_o'$$ ($\pi_o, \pi_o'$ differ in at least one cycle)

However, we need to modify this classic information flow definition to analyze a design with encryption.
In the threat model, we assume that the ciphertext after encryption will not leak any information about the plaintext. However, changing the plaintext will necessarily change the ciphertext regardless of the encryption scheme.
Therefore, we would reach the conclusion that there is information flow from the plaintext to the ciphertext, but this would be a false alert.
We describe our solution to this issue in the following subsection by using the notion of `ciphertext declassification.'

\subsection{Information Flow with Ciphertext Declassification}
The idea of \textit{declassification}~\cite{sabelfeld2009declassification, askarov2007gradual} allows  information flow under a specific condition.
This allows information flow to go through a certain variable such as the ciphertext output of the encryption unit under the condition that the encryption is finished.
The similar idea is also applied in our ISA-level proof.

Let $s$ be a secret variable, $o$ be an output variable, and $c$ be the ciphertext output of encryption. 
Let $p$ be a predicate that represents the completion of encryption, \textit{i.e.}, $p$ is only true when encryption is completed.
The intuitive way to realize declassification is to model the design such that when $p$ is true, $c$ is replaced by a free variable $c_f$.
This cuts off the connection between $s$ and $c$ when $p$ is true and blocks the information flow under this condition, but not otherwise.

Next, we will give a formal description of the above method. We construct a new FSM $SE'$ from FSM $SE$ as follows.
(i) We add a free variable $c_f$ to the input vector $I$.
(ii) For the output function $F$ and next state function $N$, replace all occurrences of $c$ in their arguments using the following expression $p\;?\;c_f:c$.

Denote the new input vector as $I'$, the new output function as $F'$, and the new next state function as $N'$.
The only difference between FSM $SE'$ and FSM $SE$ is that in $SE'$, $c$ is replaced by $c_f$ when $p$ is true.
Then, in the original FSM $SE$ there exists information flow from $s$ to $o$ not going through $c$ when $p$ holds if and only if there exists information flow from $s$ to $o$ in $SE'$. 
We call the above technique `ciphertext declassification information flow'.

\subsection{Summary}
Based on our threat model 
we want to check if there exists information flow from the secrets to the Enclave outputs not going through the ciphertext after encryption.
\rev{We will demonstrate in \S~\ref{sec:eval_overview} how both functional and timing leakage can be captured using standard hardware IFT.} 


In the next section, we \rev{detail the SE implementations, including three secure implementations and four insecure implementations, then use hardware IFT to either detect leakage or prove security.}


\section{Implementation and Evaluation}
\label{sec:evaluate}

To evaluate the effectiveness and performance of our RTL verification technique, we implemented a collection of SE Enclave designs with different microarchitectural optimizations or security flaws. The goal of the evaluation is to check if our verification technique is sufficient to support common microarchitectural design optimizations and catch security flaws when they are present. In this section, we first describe the designs we implemented as the verification targets and then explain our verification process on each of those design options, followed by experiment results.

\subsection{SE RTL Enclave Design Details}
In this subsection, we describe the details of the SE Enclave RTL designs and the flawed SE Enclave designs we used to validate our verification technique. Overall, we implemented seven different designs. Three of the designs are safe and valid designs including one default design, one area optimized design, and one design with our advanced decryption cache optimization. We also implemented four designs with different types of vulnerabilities in them to validate our verification technique.

\subsubsection{Implemented Instructions}

All our prototype designs implemented 14 instructions in total as listed in Table~\ref{table:inst}. These 14 instructions fall in categories shown in Table \ref{table:isa-overview}. Instructions in class \textit{Shift}, \textit{Arithmetic}, \textit{Logical}, and \textit{Comparison} all take two 128-bit encrypted operands as inputs and generate one 128-bit encrypted output with a fresh salt value. \texttt{ENC} instruction takes one 64-bit plaintext input and generates a 128-bit output in its encrypted form. This instruction is to support adding a public plaintext value to an encrypted secret value. To do so, the developer first encrypts the public value with a fresh 64-bit salt and then carries out the operation with the corresponding instruction. \texttt{CMOV} is the only ternary instruction supported by SE Enclave. This instruction takes three operands as input: \textit{condition}, \textit{t\_value}, and \textit{f\_value}. The output result is based on the value of the condition: \textit{t\_value} if the condition is true and \textit{f\_value} otherwise. 
\rev{For each of the designs described below, there are two outputs: \texttt{Valid} and \texttt{Data}. \texttt{Valid} is effectively a completion signal that indicates whether the value at \texttt{Data} is the valid result for the most recent instruction.}
\begin{table}[t]
\centering
\resizebox{0.68\columnwidth}{!} {
\begin{tabular}[\aboverulesep=0ex, \belowrulesep=0ex]{@{\extracolsep{0pt}}|c|c|c|c|}

\hline
Inst. & Description & Inst. & Description \\
\hline
\rowcolor[HTML]{f3f3f3}\multicolumn{2}{|c|}{Class: Shift} & \multicolumn{2}{|c|}{Class: Arithmetic}\\
\hline
\texttt{SLL} & Logic Left Shift & \texttt{ADD} & Add\\
\hline
\texttt{SLA} & Arith. Left Shift & \texttt{SUB} & Subtract\\
\hline
\texttt{SRA} & Arith. Right Shift & \texttt{MULT} & Multipliy\\
\hline
\multicolumn{2}{|c|}{\cellcolor[HTML]{f3f3f3} Class: Logical} & \texttt{MULTS} & Signed Multiply\\
\hline
\texttt{XOR} & Logical XOr & \multicolumn{2}{|c|}{\cellcolor[HTML]{f3f3f3} Class: Comparison} \\
\hline
\texttt{OR} & Logical Or & \texttt{LT} & Less Than\\
\hline
\texttt{AND} & Logical And & \texttt{LTS} &  Signed Less Than\\
\hline
\rowcolor[HTML]{f3f3f3}\multicolumn{2}{|c|}{Class: Encryption} & \multicolumn{2}{|c|}{\cellcolor[HTML]{f3f3f3} Class: Conditional}\\
\hline
\texttt{ENC} & Encrypt & \texttt{CMOV} & Conditional Move \\
\hline

\end{tabular}}
\caption{Instructions Implemented in SE Prototype. 
}
\vspace{-9mm}
\label{table:inst}
\end{table}

\subsubsection{Default SE Enclave Design (Default)}
A simplified representation of the default design is shown in Figure~\ref{fig:se_baseline}. In this design, SE Enclave takes in three operands, an instruction, and a valid bit to signal the start of computation. On the output side, SE Enclave outputs a 128-bit always encrypted value with a valid bit signaling the end of the computation. We use a 10-round unrolled AES encryption and decryption unit. The key is stored in the key register within the SE Enclave. The ALU unit is a constant time unit that performs the computation on already decrypted plaintext values. For the purpose of evaluation, we used a Linear Shift Feedback Register(LSFR) with a random seed as our random number generator(RNG), which yields a 64-bit random number per cycle. Any random number generator that can generate 64 or more random bits per cycle can be a valid design candidate. The choice and implementation of an RNG are outside the scope of this paper.

The default design can be fully pipelined. The input ciphertext is fed into the decryption unit first. After decryption, the 64-bit plaintext value goes through the ALU for computation. The completed result is padded with a newly generated 64-bit salt from RNG before being sent to the encryption unit. 

\begin{figure}[h]
    \centering
    \includegraphics[width=0.48\textwidth,trim={0 13.1cm 9.6cm 0}, clip]{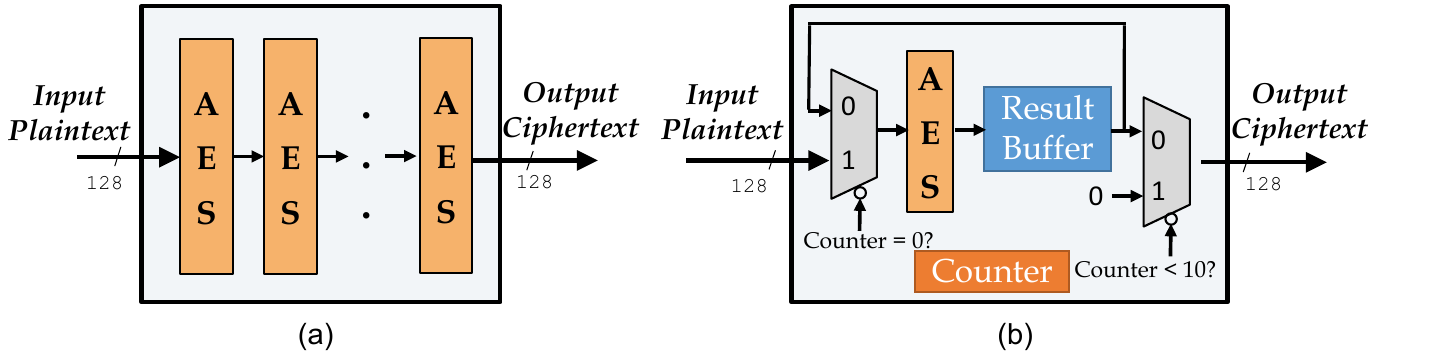}
    \vspace{-6mm}
    \caption{Unrolled AES (a) vs. Rolled AES (b)}
    \vspace{-4mm}
    \label{fig:two_aes}
\end{figure}
\subsubsection{Optimized Architecture}

We introduced two design variants of SE Enclave enabled by common optimizations: one optimized for area and the other optimized for performance.

\emph{Design with Rolled-Crypto Unit (Rolled AES)}:
One optimization 
is to roll crypto units and thus make the SE Enclave an area-optimized architecture. The rolled AES uses a single shared register for each round and stores the intermediate results in the register after the completion of each round. The rolled architecture cannot be pipelined as it has only one register to hold intermediate AES encryption/decryption results after each round, thus the upstream data would need to be blocked until the completion of the full 10 rounds of AES encryption and decryption. Figure~\ref{fig:two_aes} shows the comparison between unrolled and rolled AES encryption.

Note that it is important for a rolled AES to prevent partially encrypted ciphertext from flowing outside of the encryption module because only fully encrypted ciphertext is considered secure.

\begin{figure}[t]
    \centering
    \includegraphics[width=0.46\textwidth,trim={0 5cm 2cm 2cm},clip]{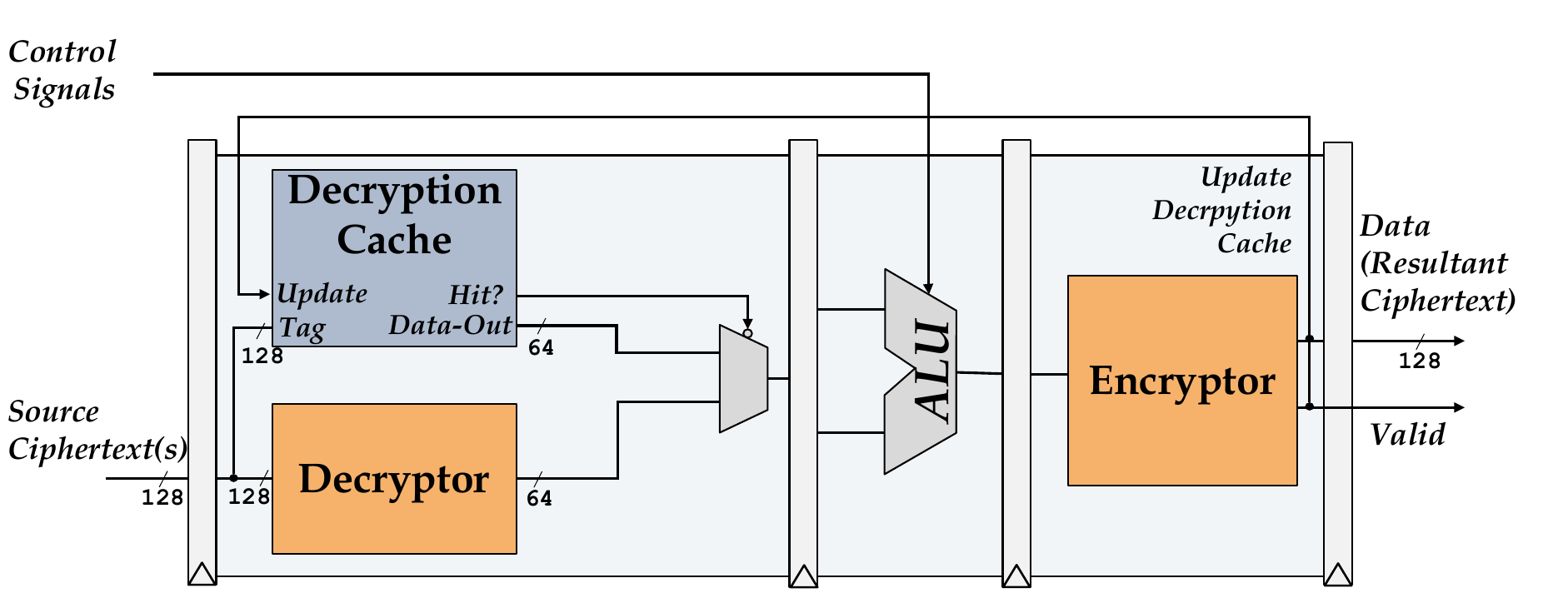}
    \vspace{-2mm}
    \caption{SE Enclave Design with Decryption Cache}
    \vspace{-5mm}
    \label{fig:se_cache}
\end{figure}
\emph{Cache-enabled SE Design (Cache)}:
Another optimization we implemented is to include a decryption cache as shown in Figure~\ref{fig:se_cache}. One observation we have is that many applications demonstrate temporal locality --- a recently computed result is more likely to be used as input operands in subsequent instructions. To exploit this type of locality, we cache the plaintext and the corresponding ciphertext for each result in the decryption cache as a First-in, First-out (FIFO) buffer. For each input operand, SE Enclave first looks up the input in the decryption cache and skips decryption if all operands are hit in the cache. We will show this type of design does not introduce side channels in the sections below.

\subsubsection{Vulnerable Designs}
To validate the effectiveness of our verification technique, we also implemented four microarchitectural vulnerabilities in various parts of the system that can leak information through side channels.

\emph{Exposed Partially Encrypted Ciphertext (Vulnerable Rolled AES)}:
For the design with the rolled-crypto unit, a designer might be tempted to connect the output of the crypto unit directly to the register that holds the partial results of the AES encryption, and thus expose the partial encryption results to an attacker who can snoop that output. 
When not properly encrypted with a sufficient number of rounds, the partially encrypted ciphertext can be easily recovered through crypto-analysis~\cite{reducedaes}. 
This vulnerable design connects the output directly to the crypto result register to emulate a common mistake that might happen in the hardware design process.
Thus, this causes a functional leakage from the SE Enclave.

\emph{Value-dependent Timing Multiplier (Vulnerable Multiplier)}:
The second vulnerability we introduced is a shift-and-add multiplier (sample code listed in Appendix~\ref{appx:example_mul}), which causes a timing leakage. 
When one of the input operands is 0, the ALU in SE Enclave can immediately output 0 as the result of a multiplication operation; otherwise, the shift-and-add multiplication takes multiple cycles (operand dependent) to complete the computation.
\rev{If the attacker measures the execution time of the \texttt{MULT} instruction for different encrypted operands, those finishing in fewest cycles indicate that one of the operands is likely to be 0, resulting in a timing leakage.}

%

\emph{Value-dependent Cache Replacement(Vulnerable Cache)}:
The third vulnerability we introduced is a value-dependent cache replacement policy, which leads to timing leakage. In this design, the cache is partitioned into two smaller caches, as shown in Figure \ref{fig:se_flawed_cache}. 
While still following FIFO, which cache the result would be placed in is dependent on the sign of the plaintext value. An attacker can conduct a known-plaintext attack~\cite{kpa,crypt-dict}, similar to \texttt{Prime-and-Probe}~\cite{prime-and-probe}, to first place a piece of data into one of the cache partitions. 
Then, the attacker replays the victim's instruction until the results would fill up one of the cache partitions. 
After that, the attacker probes the cache with the data they previously placed in the cache. 
\rev{By measuring the execution time of instructions to determine whether it is a hit or miss, the attacker can successfully recover the sign bit of victim data and even recover the full secret with enough trials through binary search as detailed in Appendix \ref{appx:attack-example}, which is similar to the approach used in other attacks, such as Blind ROP~\cite{BROP}}.

\begin{figure}[t]
    \centering
    \includegraphics[width=0.48\textwidth,trim={0 5cm 1.4cm 0},clip]{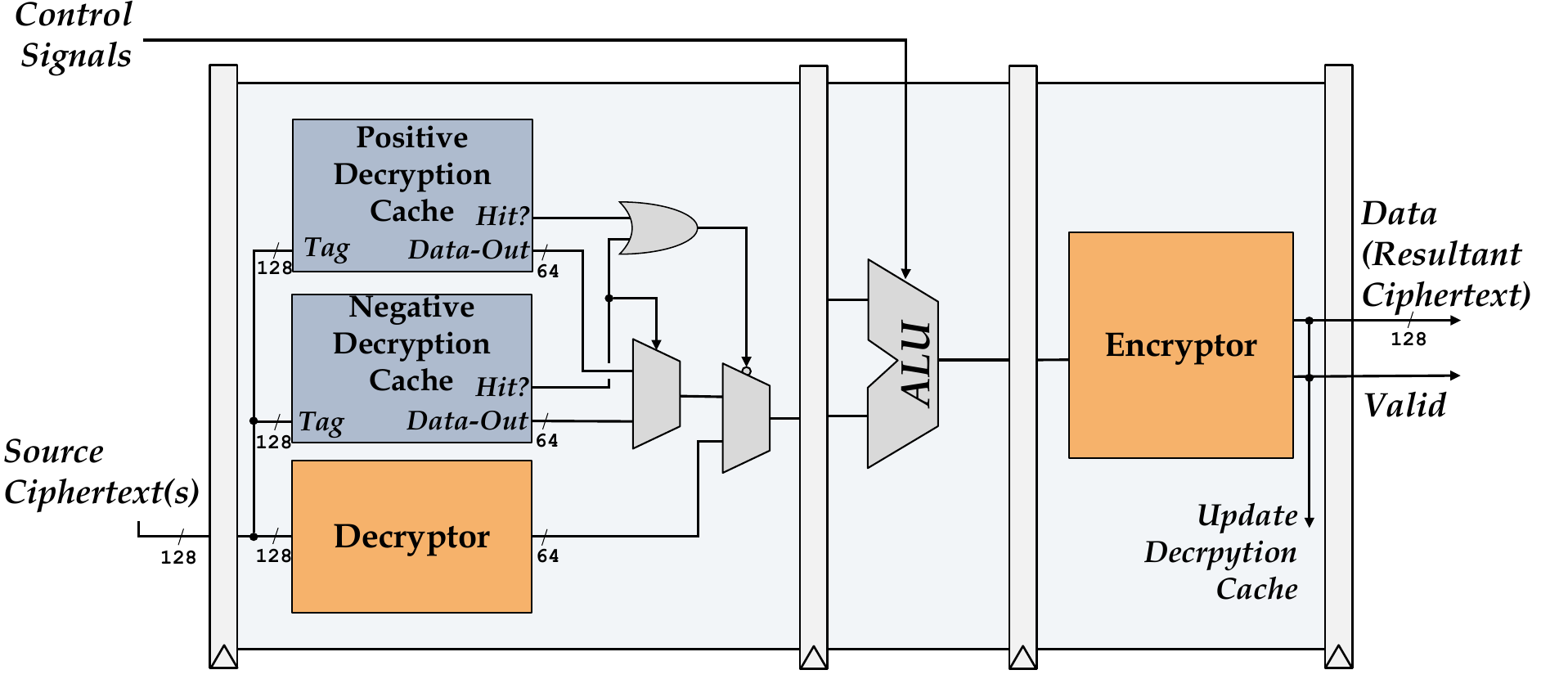}
    \vspace{-5mm}
    \caption{Flawed SE Enclave Design with Partitioned Decryption Cache}
    \vspace{-5mm}
    \label{fig:se_flawed_cache}
\end{figure}


\emph{Value-dependent Timing RSA (Vulnerable RSA)}:
Our last flawed design is a flawed RSA crypto engine (sample code listed in Appendix~\ref{appx:example_rsa}).
In RSA, to decrypt a piece of data, SE Enclave computes $m\ =\ c^d\ mod\ N$ where $m$ is the decrypted message, $c$ is the ciphertext, $d$ is decryption key, and $N$ is the modulus. 
To compute modular exponentiation effectively, the modular exponentiation module would compute bit by bit for the decryption key until it hits the most significant `1' in the key. 
\rev{This introduces a timing leakage because the decryption time depends on a secret (the private key) and the attacker can measure the execution time of instructions to gain information about the length of the key and potentially recover the whole key~\cite{wong2005timing}. }

\subsection{Evaluation Overview}
\label{sec:eval_overview}
\subsubsection{Checking RTL Properties Using IFT}
\rev{We check the RTL properties stated in Table~\ref{table:isa-rtl-mapping} using standard IFT. The secret variables (the source in IFT) are \textit{i)} the plaintext after decryption and \textit{ii)} the crypto key in the key register. The destination variables for the IFT are the two outputs of the SE Enclave, \texttt{Valid} and \texttt{Data}.
Information flow to \texttt{Valid} indicates timing leakage as when \texttt{Valid} goes high, \textit{i.e.}, the completion signal, depends on a secret. Information flow to \texttt{Data} indicates dependence of \texttt{Data} on the secret. This dependence may indicate functional leakage, \textit{i.e.}, the secret may be inferred from the value of \texttt{Data} or timing leakage depending on the sequence of values at \texttt{Data}. For example, \texttt{Data} may stay at a value of \texttt{0} till it has a valid value, and switch to this valid value when it is ready. This distinction between functional leakage and timing leakage at \texttt{Data} is made by design analysis or using techniques from previous works~\cite{oberg2014leveraging, ardeshiricham2017clepsydra}. To distinguish this from the timing leakage observed at \texttt{Valid}, we refer to timing leakage at \texttt{Data} as functional-timing leakage.
}
\rev{\subsubsection{Hardware Information Flow Tracking Tool} 
We can use any available IFT tool~\cite{website:jaspergold, website:hanna2013jasper, tiwari2009complete, ardeshiricham2017register, solt2022cellift} to check for information flow. In this work, we use Cadence JasperGold Security Path Verification (SPV) tool~\cite{website:jaspergold, website:hanna2013jasper} for its availability in both industry and academia (through a university license).}
The tool will either prove that there is no information flow from the \rev{sources (the decrypted plaintext and key register) to the destination (the SE Enclave outputs)} or find a hardware trace that demonstrates the information flow using \rev{symbolic model checking, \textit{i.e.}, over all possible inputs}.

\subsubsection{Conditional Ciphertext Declassification}
\rev{In our setting we need to check information flow with ciphertext declassification. SPV could potentially check this using \textit{blackboxing} and the \textit{not-through switch} (in previous SPV versions). SPV also allows specifying predicates on the source and destination under which information flow is allowed. However, in our setting we need to check information flow with {\it conditional ciphertext declassification} using predicates on intermediate signals, rather than with the source/destination (\textit{e.g.}, only declassify the ciphertext when the encryption is finished). 
We accomplish this with a simple modification of the RTL design as described in \S~\ref{sec:definition}.}
Note that we need to treat unrolled encryption and rolled encryption differently.
For unrolled encryption, because of its pipeline structure, we can directly replace the ciphertext output of the last crypto unit in the pipeline with a free input $c_f$ because the output of the last crypto unit is always completely encrypted which means the completion predicate $p$ is always true.
For rolled encryption, since there is only one crypto unit, we only replace its output with a free input when the counter indicates that it is the last round.
Thus, for a 10-round encryption, the completion predicate $p$ is $Counter==10$. 

\subsubsection{Cache Initialization}
\rev{When evaluating SE variants with caching, we initialize the cache to be in an arbitrary state, \textit{i.e.}, any cache line can be valid or invalid.}
\rev{Since JasperGold conducts symbolic model checking, all possible initial states will be explored, thus ensuring full coverage.}
This helps catch vulnerabilities in a much shorter time because some vulnerabilities only leak information when the cache is full.
If we initialize the cache to be empty, the formal tool needs to conduct an extremely long symbolic execution, which significantly increases the verification time.
If we initialize the cache to be full, although we may catch the vulnerability leaking information with the full cache faster, the tool cannot cover the entire search space because there is no instruction to flush the cache in the SE Enclave.
However, we will fail to capture vulnerabilities that only leak information when the cache is not full.
Therefore, initializing the cache to be in an arbitrary state can avoid the above two drawbacks.
As the experiments demonstrate, the formal tool can prove security with full coverage, and also capture vulnerabilities efficiently.


In the following subsection, we will evaluate implementations of the SE Enclave and its variants, along with several buggy implementations using the JasperGold SPV tool.

\begin{table}[ht]
\vspace{-2mm}
\centering
\resizebox{1\columnwidth}{!} {
\begin{tabular}{|c|c|c|c|c|r|}
\hline
\multirow{3}{*}{Setup} &
\multicolumn{5}{|c|}{Processor: two Intel Xeon 5222 cores} \\ 
& \multicolumn{5}{|c|}{Memory: 256 GB} \\
& \multicolumn{5}{|c|}{Tool: Cadence JasperGold 2021} \\
\hline
\hline
\rowcolor[HTML]{f3f3f3}\textbf{\ \ \ \ \ SE Variant\ \ \ \ \ } & \textbf{register bits} & \textbf{result} & \textbf{leakage} &  \textbf{time} & \rev{\textbf{memory}} \\
\hline
Default &6544 & secure & - & 0.1s & \rev{1.6GB}\\
\hline
Rolled AES&1412 &secure &- & 0.1s & \rev{0.7GB}\\
\hline 
Cache&12784&secure &- & 0.1s & \rev{1.6GB}\\
\hline 
Vulnerable &\multirow{2}{*}{1412} &\multirow{2}{*}{insecure} & functional & \multirow{2}{*}{109.4s} &\multirow{2}{*}{\rev{2.5GB}} \\
Rolled AES & & & \rev{(plaintext $\rightarrow$ 
\texttt{Data}, key $\rightarrow$ \texttt{Data})} & &\\
 \hline
Vulnerable &\multirow{2}{*}{6737} &\multirow{2}{*}{insecure} & timing & \multirow{2}{*}{63.3s} &\multirow{2}{*}{\rev{4.7GB}} \\
 Multiplier & & & \rev{(plaintext $\rightarrow$ \texttt{Valid })}& &\\
\hline 
Vulnerable & \multirow{2}{*}{12752} &\multirow{2}{*}{insecure} & timing & \multirow{2}{*}{402.4s} &\multirow{2}{*}{\rev{14.7GB}} \\
 Cache & & & \rev{(plaintext $\rightarrow$ \texttt{Valid })}& &\\
\hline 
Vulnerable &\multirow{2}{*}{4328} &\multirow{2}{*}{insecure} & \rev{timing, functional-timing}  & \multirow{2}{*}{0.1s} & \multirow{2}{*}{\rev{0.3GB}}\\
 RSA & & & \rev{(key $\rightarrow$ \texttt{Valid}, key $\rightarrow$ \texttt{Data} )} & &\\
\hline 
\end{tabular}}
\vspace{-0.5mm}
\caption{Experimental Evaluation}
\vspace{-9mm}
\label{table:exp_result}
\end{table}

\subsection{Experimental Results}
The experimental results are provided in Table~\ref{table:exp_result}.
For every SE design variant, we show the number of register bits as an indicator of the size of the state space.
\rev{In addition to the verification result, \textit{i.e.}, \emph{secure} or \emph{insecure}, we also label every insecure design with the leakage type, along with the source and sink of the information flow captured.
The \texttt{Vulnerable Rolled AES} implementation has functional leakage where both plaintext secrets and the encryption key are leaked to \texttt{Data}.
The \texttt{Vulnerable Multiplier} and \texttt{Vulnerable Cache} have timing leakage where the plaintext secrets are leaked to \texttt{Valid}.
For \texttt{Vulnerable RSA}, there exists information flow to both \texttt{Data} and \texttt{Valid}.
The leakage through \texttt{Valid} is timing leakage, but there is also timing leakage through \texttt{Data} as it is also leaking information about execution time as follows. 
The RSA crypto engine can only be implemented as a non-pipeline structure (the number of rounds is variable and depends on the decryption key).
Thus, similar to rolled AES encryption, its output needs to be blocked when the encryption is ongoing.
This means that when the ciphertext is not ready, it will output some default value such as all \texttt{0} to the \texttt{Data}.
The attacker can infer the decryption key by measuring the number of \texttt{0}s between two non-zero ciphertext outputs which is timing leakage. As discussed in \S~\ref{sec:eval_overview}, we distinguish between the timing leakage to \texttt{Valid} and \texttt{Data}  by referring to the latter as functional-timing leakage.

In a pipelined AES encryption in the default SE Enclave design, the crypto engine outputs a fully encrypted ciphertext every cycle, thus only \texttt{Valid} carries timing information and this attack does not work.
This explains why there is no information flow to \texttt{Data} for \texttt{Vulnerable Multiplier} and \texttt{Vulnerable Cache}.
}

\rev{ As shown in Table~\ref{table:exp_result}, it takes less time and memory to prove no information leakage for secure implementations than to detect information leakage for the vulnerable ones.
In secure designs, since the path from secrets to both outputs is cut off after declassification, SPV only needs to do a simple structural path check to prove the security. In comparison, SPV needs to do a state space search in order to detect information leakage in a vulnerable design.
The time and memory usage are also affected by the number of register bits and the complexity of the design.}
Across all implementations, the maximum verification time is less than 7 minutes.
The experimental results demonstrate that our evaluation scheme is able to prove security or catch leakage precisely and efficiently.
\vspace{-1mm}
\subsection{Summary}
In this section, we evaluated different SE Enclave variants using information flow checking.
In general, our information flow definition and evaluation methodology work on any low-trust architecture with encryption.
This is because \rev{low-trust architectures limit trust to only a small hardware enclave which facilitates formal verification, and the conditional ciphertext declassification we implemented can correctly deal with the information flow going through encryption.
}
\section{Related Work} \label{sec:related}
\subsection{Secure Hardware Architectures}
There have been numerous works in designing secure architectures as microarchitectural flaws (and lack of security awareness at the architecture level) continue to be exploited by software attacks.

Trusted Execution Environments (TEEs)~\cite{champagne2010scalable, suh2003aegis, suh2005design, alves2004trustzone, costan2016sanctum, bourgeat2019mi6, fletcher2012secure, ren2017design, anati2013innovative, costan2016intel, mckeen2013innovative, schunter2016intel, lee2020keystone} have been widely deployed by mainstream hardware vendors in their server-grade CPUs to provide execution integrity and data confidentiality. 
\rev{Although the root of trust is also in hardware, TEEs such as Intel SGX~\cite{costan2016intel} and Keystone~\cite{lee2020keystone} have a significantly different threat model as they do not eliminate timing side channels.}
Sanctum~\cite{costan2016sanctum}, MI6~\cite{bourgeat2019mi6}, and Ascend~\cite{fletcher2012secure, ren2017design} aim to eliminate timing side channels, but they have much larger trusted computing bases than SE, making them hard to be formally verified.

Multiple works have been developed to mitigate leakage of data through timing side channels during speculative execution. 
Hyperflow~\cite{ferraiuolo_hyperflow_2018} enforces data security properties by static flow analysis at the hardware construction time. It relies on a trusted label manager to assign correct labels to data.  OISA~\cite{yu_data_2019} presents a timing-channel free ISA extension that uses tags to distinguish \textit{Public} and \textit{Confidential} data. It verifies both ISA and microarchitectural implementation through an abstract machine. Speculative Taint Tracking (STT)~\cite{stt} and Speculative Data-Oblivious Execution~(SDO)~\cite{sdo} are consecutive works that use runtime speculative taint analysis to eliminate timing side channels for safe speculative execution by delaying transient execution on the backend. DOLMA~\cite{dolma} introduces the principle of transient non-observability and delays executions that are reliant on speculative memory micro-ops. Software-hardware contract~\cite{sw-hw-contract} formalizes software security requirements and hardware capabilities to ensure confidentiality for benign programs. 

In comparison, we present the ISA and RTL level verification of a low-trust architecture that cannot leak secrets for any (even malicious) program, which is absent in the aforementioned works.

\rev{\subsection{Functional Leakage and Timing Leakage}

In this work, we use the role of signals to differentiate between functional leakage and timing leakage, \textit{i.e.}, using hardware information flow tracking at \texttt{Data} for functional leakage and at \texttt{Valid} for timing leakage.
This idea was discussed in a previous work~\cite{hu2021hardware}.

Aside from using the role of signals, previous works also tried to use the observation of a signal in consecutive cycles (sequence) to differentiate between functional and timing leakage~\cite{oberg2014leveraging, ardeshiricham2017clepsydra}.
Such techniques require more complex hardware information flow tracking logic and can separate timing leakage and functional leakage without needing to know the role of signals.
However, since our work is interested in capturing both types of leakage, we choose to use the simpler role-based classification.}

\subsection{Hardware Security Evaluation Schemes}

\subsubsection{Type-Based Hardware Security Evaluation}
SecVerilog~\cite{zhang2015hardware} and its variants~\cite{ferraiuolo2017verification, ferraiuolo2017secure} introduce new hardware design languages that allow developers to attach different security levels to hardware variables while programming, and also to define the rules for information flowing between different security levels.
To declassify, these languages provide 'downgrading' syntax in the language to permit specific information flows.
Static type checking is conducted to formally verify that the security policy is followed by the design.

\rev{While it is computationally faster, verification based on static type checking is not as precise as verification based on symbolic checking as in our work~\cite{khoo2010mixing}.}
Static type checking may easily produce false positives due to its conservative nature.
Further, it is inconvenient and error-prone for users to use a new language along with security labels in developing hardware.
\subsubsection{Other Evaluation Schemes Based on Information Flow}
Depending on threat models, different security properties have been proposed and verified.
These properties are based on the basic information flow property and add different conditions and constraints according to the threat model.
Attempts have been made to detect timing leakage with taint propagation~\cite{gleissenthall2019iodine, oberg2014leveraging} by verifying properties such as `constant-time execution'. Works based on Unique Program Execution Checking~\cite{fadiheh2020formal, fadiheh2019processor} are aimed at checking information leakage through out-of-order execution, where information flow caused by in-order execution is ruled out by constraints.

A key advantage of our work is that we are able to verify security properties at the instruction level for all programs while other works do not provide the same guarantees.

\section{Conclusions}\label{sec:conclusions}

This work demonstrates how the security of all programs running on low-trust architectures can be ensured by clearly defining the security requirements of their ISA instructions, formally specifying the consequent proof obligations for RTL implementations, and then executing these proof obligations using RTL formal verification tools. Further, it shows how these proof obligations cover functional as well as timing side-channel leakage. Finally, the small footprint of the trusted part of the implementation enables completing the formal checks using existing state-of-the-art formal verification tools - something that is currently not possible for non-low-trust architectures. We demonstrate our approach using the SE architecture where the ISA and the program-level proof are handwritten and the RTL-level verification done using off-the-shelf formal tools.
Our experiments using seven different design variants shows that our approach is effective in proving the security of correct implementations and detecting flaws in buggy implementations. 
}

\begin{acks}
We thank Professor Paul Grubbs from the University of Michigan for lending his expertise in cryptography.
\end{acks}

\bibliographystyle{ACM-Reference-Format}
\bibliography{refs}


\begin{thebibliography}{85}


\ifx \showCODEN    \undefined \def \showCODEN     #1{\unskip}     \fi
\ifx \showDOI      \undefined \def \showDOI       #1{#1}\fi
\ifx \showISBNx    \undefined \def \showISBNx     #1{\unskip}     \fi
\ifx \showISBNxiii \undefined \def \showISBNxiii  #1{\unskip}     \fi
\ifx \showISSN     \undefined \def \showISSN      #1{\unskip}     \fi
\ifx \showLCCN     \undefined \def \showLCCN      #1{\unskip}     \fi
\ifx \shownote     \undefined \def \shownote      #1{#1}          \fi
\ifx \showarticletitle \undefined \def \showarticletitle #1{#1}   \fi
\ifx \showURL      \undefined \def \showURL       {\relax}        \fi
\providecommand\bibfield[2]{#2}
\providecommand\bibinfo[2]{#2}
\providecommand\natexlab[1]{#1}
\providecommand\showeprint[2][]{arXiv:#2}

\bibitem[kpa(2023)]%
        {kpa}
 \bibinfo{year}{2023}\natexlab{}.
\newblock \showarticletitle{Known-plaintext Attack.}
\newblock  (\bibinfo{year}{2023}).
\newblock
\urldef\tempurl%
\url{https://en.wikipedia.org/wiki/Known-plaintext_attack}
\showURL{%
\tempurl}


\bibitem[Alam et~al\mbox{.}(2018)]%
        {oneanddone}
\bibfield{author}{\bibinfo{person}{Monjur Alam}, \bibinfo{person}{Haider~Adnan
  Khan}, \bibinfo{person}{Moumita Dey}, \bibinfo{person}{Nishith Sinha},
  \bibinfo{person}{Robert Callan}, \bibinfo{person}{Alenka Zajic}, {and}
  \bibinfo{person}{Milos Prvulovic}.} \bibinfo{year}{2018}\natexlab{}.
\newblock \showarticletitle{{One\&Done}: A {Single-Decryption} {EM-Based}
  Attack on {OpenSSL{\textquoteright}s} {Constant-Time} Blinded {RSA}}. In
  \bibinfo{booktitle}{\emph{27th USENIX Security Symposium (USENIX Security
  18)}}. \bibinfo{publisher}{USENIX Association}, \bibinfo{address}{Baltimore,
  MD}, \bibinfo{pages}{585--602}.
\newblock
\showISBNx{978-1-939133-04-5}
\urldef\tempurl%
\url{https://www.usenix.org/conference/usenixsecurity18/presentation/alam}
\showURL{%
\tempurl}


\bibitem[Alves(2004)]%
        {alves2004trustzone}
\bibfield{author}{\bibinfo{person}{Tiago Alves}.}
  \bibinfo{year}{2004}\natexlab{}.
\newblock \showarticletitle{Trustzone: Integrated Hardware and Software
  Security
  \url{https://www.techonline.com/tech-papers/trustzone-integrated-hardware-and-software-security/}}.
\newblock \bibinfo{journal}{\emph{White paper}} (\bibinfo{year}{2004}).
\newblock
\urldef\tempurl%
\url{https://www.techonline.com/tech-papers/trustzone-integrated-hardware-and-software-security/}
\showURL{%
\tempurl}


\bibitem[Anati et~al\mbox{.}(2013)]%
        {anati2013innovative}
\bibfield{author}{\bibinfo{person}{Ittai Anati}, \bibinfo{person}{Shay Gueron},
  \bibinfo{person}{Simon Johnson}, {and} \bibinfo{person}{Vincent Scarlata}.}
  \bibinfo{year}{2013}\natexlab{}.
\newblock \showarticletitle{Innovative Technology for CPU based Attestation and
  Sealing}. In \bibinfo{booktitle}{\emph{Proceedings of the 2nd international
  workshop on hardware and architectural support for security and privacy}},
  Vol.~\bibinfo{volume}{13}. ACM New York, NY, USA.
\newblock


\bibitem[Ardeshiricham et~al\mbox{.}(2017a)]%
        {ardeshiricham2017clepsydra}
\bibfield{author}{\bibinfo{person}{Armaiti Ardeshiricham}, \bibinfo{person}{Wei
  Hu}, {and} \bibinfo{person}{Ryan Kastner}.} \bibinfo{year}{2017}\natexlab{a}.
\newblock \showarticletitle{Clepsydra: Modeling timing flows in hardware
  designs}. In \bibinfo{booktitle}{\emph{2017 IEEE/ACM International Conference
  on Computer-Aided Design (ICCAD)}}. IEEE, \bibinfo{pages}{147--154}.
\newblock


\bibitem[Ardeshiricham et~al\mbox{.}(2017b)]%
        {ardeshiricham2017register}
\bibfield{author}{\bibinfo{person}{Armaiti Ardeshiricham}, \bibinfo{person}{Wei
  Hu}, \bibinfo{person}{Joshua Marxen}, {and} \bibinfo{person}{Ryan Kastner}.}
  \bibinfo{year}{2017}\natexlab{b}.
\newblock \showarticletitle{Register transfer level information flow tracking
  for provably secure hardware design}. In \bibinfo{booktitle}{\emph{Design,
  Automation \& Test in Europe Conference \& Exhibition (DATE), 2017}}. IEEE,
  \bibinfo{pages}{1691--1696}.
\newblock


\bibitem[Askarov and Sabelfeld(2007)]%
        {askarov2007gradual}
\bibfield{author}{\bibinfo{person}{Aslan Askarov} {and} \bibinfo{person}{Andrei
  Sabelfeld}.} \bibinfo{year}{2007}\natexlab{}.
\newblock \showarticletitle{Gradual Release: Unifying Declassification,
  Encryption and Key Release Policies}. In \bibinfo{booktitle}{\emph{2007 IEEE
  Symposium on Security and Privacy (SP'07)}}. IEEE, \bibinfo{pages}{207--221}.
\newblock


\bibitem[Attridge(2002)]%
        {attridge2002overview}
\bibfield{author}{\bibinfo{person}{Jim Attridge}.}
  \bibinfo{year}{2002}\natexlab{}.
\newblock \showarticletitle{An Overview of Hardware Security Modules}.
\newblock \bibinfo{journal}{\emph{SANS Institute. Information Security Reading
  Room}} (\bibinfo{year}{2002}).
\newblock


\bibitem[Avanzi(2017)]%
        {avanzi2017qarma}
\bibfield{author}{\bibinfo{person}{Roberto Avanzi}.}
  \bibinfo{year}{2017}\natexlab{}.
\newblock \showarticletitle{The QARMA Block Cipher Family. Almost MDS Matrices
  Over Rings With Zero Divisors, Nearly Symmetric Even-Mansour Constructions
  With Non-Involutory Central Rounds, and Search Heuristics for Low-Latency
  S-Boxes}.
\newblock \bibinfo{journal}{\emph{IACR Transactions on Symmetric Cryptology}}
  \bibinfo{volume}{2017}, \bibinfo{number}{1} (\bibinfo{date}{Mar.}
  \bibinfo{year}{2017}), \bibinfo{pages}{4–44}.
\newblock
\urldef\tempurl%
\url{https://doi.org/10.13154/tosc.v2017.i1.4-44}
\showDOI{\tempurl}


\bibitem[Bardeh and R{\o}njom(2019)]%
        {reducedaes}
\bibfield{author}{\bibinfo{person}{Navid~Ghaedi Bardeh} {and}
  \bibinfo{person}{Sondre R{\o}njom}.} \bibinfo{year}{2019}\natexlab{}.
\newblock \showarticletitle{Practical Attacks on Reduced-Round AES}. In
  \bibinfo{booktitle}{\emph{Progress in Cryptology -- AFRICACRYPT 2019}},
  \bibfield{editor}{\bibinfo{person}{Johannes Buchmann},
  \bibinfo{person}{Abderrahmane Nitaj}, {and} \bibinfo{person}{Tajjeeddine
  Rachidi}} (Eds.). \bibinfo{publisher}{Springer International Publishing},
  \bibinfo{address}{Cham}, \bibinfo{pages}{297--310}.
\newblock
\showISBNx{978-3-030-23696-0}


\bibitem[Baudet et~al\mbox{.}(2011)]%
        {prng}
\bibfield{author}{\bibinfo{person}{Mathieu Baudet}, \bibinfo{person}{David
  Lubicz}, \bibinfo{person}{Julien Micolod}, {and} \bibinfo{person}{Andr\'{e}
  Tassiaux}.} \bibinfo{year}{2011}\natexlab{}.
\newblock \showarticletitle{On the Security of Oscillator-Based Random Number
  Generators}.
\newblock \bibinfo{journal}{\emph{J. Cryptol.}} \bibinfo{volume}{24},
  \bibinfo{number}{2} (\bibinfo{date}{apr} \bibinfo{year}{2011}),
  \bibinfo{pages}{398–425}.
\newblock
\showISSN{0933-2790}
\urldef\tempurl%
\url{https://doi.org/10.1007/s00145-010-9089-3}
\showDOI{\tempurl}


\bibitem[Beaulieu et~al\mbox{.}(2015)]%
        {7167361}
\bibfield{author}{\bibinfo{person}{Ray Beaulieu}, \bibinfo{person}{Stefan
  Treatman-Clark}, \bibinfo{person}{Douglas Shors}, \bibinfo{person}{Bryan
  Weeks}, \bibinfo{person}{Jason Smith}, {and} \bibinfo{person}{Louis
  Wingers}.} \bibinfo{year}{2015}\natexlab{}.
\newblock \showarticletitle{The SIMON and SPECK Lightweight Block Ciphers}. In
  \bibinfo{booktitle}{\emph{2015 52nd ACM/EDAC/IEEE Design Automation
  Conference (DAC)}}. \bibinfo{pages}{1--6}.
\newblock
\urldef\tempurl%
\url{https://doi.org/10.1145/2744769.2747946}
\showDOI{\tempurl}


\bibitem[Biernacki et~al\mbox{.}(2022)]%
        {se-paper}
\bibfield{author}{\bibinfo{person}{Lauren Biernacki},
  \bibinfo{person}{Meron~Zerihun Demissie}, \bibinfo{person}{Kidus~Birkayehu
  Workneh}, \bibinfo{person}{Fitsum~Assamnew Andargie}, {and}
  \bibinfo{person}{Todd Austin}.} \bibinfo{year}{2022}\natexlab{}.
\newblock \showarticletitle{Sequestered Encryption: A Hardware Technique for
  Comprehensive Data Privacy}. In \bibinfo{booktitle}{\emph{2022 IEEE
  International Symposium on Secure and Private Execution Environment Design
  (SEED)}}. \bibinfo{pages}{73--84}.
\newblock
\urldef\tempurl%
\url{https://doi.org/10.1109/SEED55351.2022.00014}
\showDOI{\tempurl}


\bibitem[Bittau et~al\mbox{.}(2014)]%
        {BROP}
\bibfield{author}{\bibinfo{person}{Andrea Bittau}, \bibinfo{person}{Adam
  Belay}, \bibinfo{person}{Ali Mashtizadeh}, \bibinfo{person}{David
  Mazi{\`e}res}, {and} \bibinfo{person}{Dan Boneh}.}
  \bibinfo{year}{2014}\natexlab{}.
\newblock \showarticletitle{Hacking Blind}. In \bibinfo{booktitle}{\emph{2014
  IEEE Symposium on Security and Privacy}}. IEEE, \bibinfo{pages}{227--242}.
\newblock


\bibitem[Blaze et~al\mbox{.}(1998)]%
        {blaze1998keynote}
\bibfield{author}{\bibinfo{person}{Matt Blaze}, \bibinfo{person}{Joan
  Feigenbaum}, {and} \bibinfo{person}{Angelos~D Keromytis}.}
  \bibinfo{year}{1998}\natexlab{}.
\newblock \showarticletitle{KeyNote: Trust Management for Public-Key
  Infrastructures}. In \bibinfo{booktitle}{\emph{International Workshop on
  Security Protocols}}. Springer, \bibinfo{pages}{59--63}.
\newblock


\bibitem[Bourgeat et~al\mbox{.}(2019)]%
        {bourgeat2019mi6}
\bibfield{author}{\bibinfo{person}{Thomas Bourgeat}, \bibinfo{person}{Ilia
  Lebedev}, \bibinfo{person}{Andrew Wright}, \bibinfo{person}{Sizhuo Zhang},
  {and} \bibinfo{person}{Srinivas Devadas}.} \bibinfo{year}{2019}\natexlab{}.
\newblock \showarticletitle{Mi6: Secure Enclaves in a Speculative Out-Of-Order
  Processor}. In \bibinfo{booktitle}{\emph{Proceedings of the 52nd Annual
  IEEE/ACM International Symposium on Microarchitecture}}.
  \bibinfo{pages}{42--56}.
\newblock


\bibitem[Bresson et~al\mbox{.}(2007)]%
        {bresson2007provably}
\bibfield{author}{\bibinfo{person}{Emmanuel Bresson}, \bibinfo{person}{Olivier
  Chevassut}, {and} \bibinfo{person}{David Pointcheval}.}
  \bibinfo{year}{2007}\natexlab{}.
\newblock \showarticletitle{Provably Secure Authenticated Group Diffie-Hellman
  Key Exchange}.
\newblock \bibinfo{journal}{\emph{ACM Transactions on Information and System
  Security (TISSEC)}} \bibinfo{volume}{10}, \bibinfo{number}{3}
  (\bibinfo{year}{2007}), \bibinfo{pages}{10--es}.
\newblock


\bibitem[Brown(2018)]%
        {brown2018dieharder}
\bibfield{author}{\bibinfo{person}{Robert~G Brown}.}
  \bibinfo{year}{2018}\natexlab{}.
\newblock \showarticletitle{Dieharder: A Random Number Test Suite}.
\newblock  (\bibinfo{year}{2018}).
\newblock
\urldef\tempurl%
\url{https://webhome.phy.duke.edu/~rgb/General/dieharder.php}
\showURL{%
\tempurl}


\bibitem[Cadence(2022)]%
        {website:jaspergold}
\bibfield{author}{\bibinfo{person}{Cadence}.} \bibinfo{year}{2022}\natexlab{}.
\newblock \bibinfo{title}{JasperGold Security Path Verification.}
\newblock
\newblock
\urldef\tempurl%
\url{https://www.cadence.com/en_US/home/tools/system-design-and-verification/formal-and-static-verification/jasper-gold-verification-platform/security-path-verification-app.html}
\showURL{%
\tempurl}


\bibitem[Champagne and Lee(2010)]%
        {champagne2010scalable}
\bibfield{author}{\bibinfo{person}{David Champagne} {and}
  \bibinfo{person}{Ruby~B Lee}.} \bibinfo{year}{2010}\natexlab{}.
\newblock \showarticletitle{Scalable Architectural Support for Trusted
  Software}. In \bibinfo{booktitle}{\emph{HPCA-16 2010 The Sixteenth
  International Symposium on High-Performance Computer Architecture}}. IEEE,
  \bibinfo{pages}{1--12}.
\newblock


\bibitem[Clarkson and Schneider(2010)]%
        {clarkson2010hyperproperties}
\bibfield{author}{\bibinfo{person}{Michael~R Clarkson} {and}
  \bibinfo{person}{Fred~B Schneider}.} \bibinfo{year}{2010}\natexlab{}.
\newblock \showarticletitle{Hyperproperties}.
\newblock \bibinfo{journal}{\emph{Journal of Computer Security}}
  \bibinfo{volume}{18}, \bibinfo{number}{6} (\bibinfo{year}{2010}),
  \bibinfo{pages}{1157--1210}.
\newblock


\bibitem[Costan and Devadas(2016)]%
        {costan2016intel}
\bibfield{author}{\bibinfo{person}{Victor Costan} {and}
  \bibinfo{person}{Srinivas Devadas}.} \bibinfo{year}{2016}\natexlab{}.
\newblock \showarticletitle{Intel SGX Explained}.
\newblock \bibinfo{journal}{\emph{Cryptology ePrint Archive}}
  (\bibinfo{year}{2016}).
\newblock


\bibitem[Costan et~al\mbox{.}(2016)]%
        {costan2016sanctum}
\bibfield{author}{\bibinfo{person}{Victor Costan}, \bibinfo{person}{Ilia
  Lebedev}, {and} \bibinfo{person}{Srinivas Devadas}.}
  \bibinfo{year}{2016}\natexlab{}.
\newblock \showarticletitle{Sanctum: Minimal Hardware Extensions for Strong
  Software Isolation}. In \bibinfo{booktitle}{\emph{25th USENIX Security
  Symposium (USENIX Security 16)}}. \bibinfo{pages}{857--874}.
\newblock


\bibitem[Delgado-Restituto et~al\mbox{.}(1992)]%
        {chaoticswitch}
\bibfield{author}{\bibinfo{person}{M. Delgado-Restituto}, \bibinfo{person}{A.
  Rodriguez-Vasquez}, \bibinfo{person}{S. Espejo}, {and} \bibinfo{person}{J.L.
  Huertas}.} \bibinfo{year}{1992}\natexlab{}.
\newblock \showarticletitle{A Chaotic Switched-Capacitor Circuit for 1/f Noise
  Generation}.
\newblock \bibinfo{journal}{\emph{IEEE Transactions on Circuits and Systems I:
  Fundamental Theory and Applications}} \bibinfo{volume}{39},
  \bibinfo{number}{4} (\bibinfo{year}{1992}), \bibinfo{pages}{325--328}.
\newblock
\urldef\tempurl%
\url{https://doi.org/10.1109/81.129465}
\showDOI{\tempurl}


\bibitem[Fadiheh et~al\mbox{.}(2020)]%
        {fadiheh2020formal}
\bibfield{author}{\bibinfo{person}{Mohammad~Rahmani Fadiheh},
  \bibinfo{person}{Johannes M{\"u}ller}, \bibinfo{person}{Raik Brinkmann},
  \bibinfo{person}{Subhasish Mitra}, \bibinfo{person}{Dominik Stoffel}, {and}
  \bibinfo{person}{Wolfgang Kunz}.} \bibinfo{year}{2020}\natexlab{}.
\newblock \showarticletitle{A Formal Approach for Detecting Vulnerabilities to
  Transient Execution Attacks in Out-of-Order Processors}. In
  \bibinfo{booktitle}{\emph{2020 57th ACM/IEEE Design Automation Conference
  (DAC)}}. IEEE, \bibinfo{pages}{1--6}.
\newblock


\bibitem[Fadiheh et~al\mbox{.}(2019)]%
        {fadiheh2019processor}
\bibfield{author}{\bibinfo{person}{Mohammad~Rahmani Fadiheh},
  \bibinfo{person}{Dominik Stoffel}, \bibinfo{person}{Clark Barrett},
  \bibinfo{person}{Subhasish Mitra}, {and} \bibinfo{person}{Wolfgang Kunz}.}
  \bibinfo{year}{2019}\natexlab{}.
\newblock \showarticletitle{Processor Hardware Security Vulnerabilities and
  Their Detection by Unique Program Execution Checking}. In
  \bibinfo{booktitle}{\emph{2019 Design, Automation \& Test in Europe
  Conference \& Exhibition (DATE)}}. IEEE, \bibinfo{pages}{994--999}.
\newblock


\bibitem[Ferraiuolo et~al\mbox{.}(2017a)]%
        {ferraiuolo2017secure}
\bibfield{author}{\bibinfo{person}{Andrew Ferraiuolo}, \bibinfo{person}{Weizhe
  Hua}, \bibinfo{person}{Andrew~C Myers}, {and} \bibinfo{person}{G~Edward
  Suh}.} \bibinfo{year}{2017}\natexlab{a}.
\newblock \showarticletitle{Secure Information Flow Verification with Mutable
  Dependent Types}. In \bibinfo{booktitle}{\emph{2017 54th ACM/EDAC/IEEE Design
  Automation Conference (DAC)}}. IEEE, \bibinfo{pages}{1--6}.
\newblock


\bibitem[Ferraiuolo et~al\mbox{.}(2017b)]%
        {ferraiuolo2017verification}
\bibfield{author}{\bibinfo{person}{Andrew Ferraiuolo}, \bibinfo{person}{Rui
  Xu}, \bibinfo{person}{Danfeng Zhang}, \bibinfo{person}{Andrew~C Myers}, {and}
  \bibinfo{person}{G~Edward Suh}.} \bibinfo{year}{2017}\natexlab{b}.
\newblock \showarticletitle{Verification of a Practical Hardware Security
  Architecture through Static Information Flow Analysis}. In
  \bibinfo{booktitle}{\emph{Proceedings of the Twenty-Second International
  Conference on Architectural Support for Programming Languages and Operating
  Systems}}. \bibinfo{pages}{555--568}.
\newblock


\bibitem[Ferraiuolo et~al\mbox{.}({[n.\,d.]})]%
        {ferraiuolo_hyperflow_2018}
\bibfield{author}{\bibinfo{person}{Andrew Ferraiuolo}, \bibinfo{person}{Mark
  Zhao}, \bibinfo{person}{Andrew~C. Myers}, {and} \bibinfo{person}{G.~Edward
  Suh}.} \bibinfo{year}{[n.\,d.]}\natexlab{}.
\newblock \showarticletitle{{HyperFlow}: A Processor Architecture for
  Nonmalleable, Timing-Safe Information Flow Security}. In
  \bibinfo{booktitle}{\emph{Proceedings of the 2018 {ACM} {SIGSAC} Conference
  on Computer and Communications Security}} (Toronto Canada, 2018-10-15).
  \bibinfo{publisher}{{ACM}}, \bibinfo{pages}{1583--1600}.
\newblock
\showISBNx{978-1-4503-5693-0}
\urldef\tempurl%
\url{https://doi.org/10.1145/3243734.3243743}
\showDOI{\tempurl}


\bibitem[Fletcher et~al\mbox{.}(2012)]%
        {fletcher2012secure}
\bibfield{author}{\bibinfo{person}{Christopher~W Fletcher},
  \bibinfo{person}{Marten~van Dijk}, {and} \bibinfo{person}{Srinivas Devadas}.}
  \bibinfo{year}{2012}\natexlab{}.
\newblock \showarticletitle{A Secure Processor Architecture for Encrypted
  Computation on Untrusted Programs}. In \bibinfo{booktitle}{\emph{Proceedings
  of the seventh ACM workshop on Scalable trusted computing}}.
  \bibinfo{pages}{3--8}.
\newblock


\bibitem[Gleissenthall et~al\mbox{.}(2019)]%
        {gleissenthall2019iodine}
\bibfield{author}{\bibinfo{person}{Klaus~v Gleissenthall},
  \bibinfo{person}{Rami~G{\"o}khan K{\i}c{\i}}, \bibinfo{person}{Deian Stefan},
  {and} \bibinfo{person}{Ranjit Jhala}.} \bibinfo{year}{2019}\natexlab{}.
\newblock \showarticletitle{IODINE: Verifying Constant-Time Execution of
  Hardware}. In \bibinfo{booktitle}{\emph{28th USENIX Security Symposium
  (USENIX Security 19)}}. \bibinfo{pages}{1411--1428}.
\newblock


\bibitem[Goguen and Meseguer(1982)]%
        {goguen1982security}
\bibfield{author}{\bibinfo{person}{Joseph~A Goguen} {and}
  \bibinfo{person}{Jos{\'e} Meseguer}.} \bibinfo{year}{1982}\natexlab{}.
\newblock \showarticletitle{Security Policies and Security Models}. In
  \bibinfo{booktitle}{\emph{1982 IEEE Symposium on Security and Privacy}}.
  IEEE, \bibinfo{pages}{11--11}.
\newblock


\bibitem[Goldwasser and Micali(1982)]%
        {semnatically-secure}
\bibfield{author}{\bibinfo{person}{Shafi Goldwasser} {and}
  \bibinfo{person}{Silvio Micali}.} \bibinfo{year}{1982}\natexlab{}.
\newblock \showarticletitle{Probabilistic Encryption: How to Play Mental Poker
  Keeping Secret All Partial Information}. In
  \bibinfo{booktitle}{\emph{Proceedings of the Fourteenth Annual ACM Symposium
  on Theory of Computing}} (San Francisco, California, USA)
  \emph{(\bibinfo{series}{STOC '82})}. \bibinfo{publisher}{Association for
  Computing Machinery}, \bibinfo{address}{New York, NY, USA},
  \bibinfo{pages}{365–377}.
\newblock
\showISBNx{0897910702}
\urldef\tempurl%
\url{https://doi.org/10.1145/800070.802212}
\showDOI{\tempurl}


\bibitem[Goubin and Patarin(1999)]%
        {goubin1999and}
\bibfield{author}{\bibinfo{person}{Louis Goubin} {and} \bibinfo{person}{Jacques
  Patarin}.} \bibinfo{year}{1999}\natexlab{}.
\newblock \showarticletitle{DES and Differential Power Analysis the
  “Duplication” Method}. In \bibinfo{booktitle}{\emph{Cryptographic
  Hardware and Embedded Systems: First InternationalWorkshop, CHES’99
  Worcester, MA, USA, August 12--13, 1999 Proceedings 1}}. Springer,
  \bibinfo{pages}{158--172}.
\newblock


\bibitem[Gro{\ss}sch{\"a}dl et~al\mbox{.}(2009)]%
        {grossschadl2009side}
\bibfield{author}{\bibinfo{person}{Johann Gro{\ss}sch{\"a}dl},
  \bibinfo{person}{Elisabeth Oswald}, \bibinfo{person}{Dan Page}, {and}
  \bibinfo{person}{Michael Tunstall}.} \bibinfo{year}{2009}\natexlab{}.
\newblock \showarticletitle{Side-channel Analysis of Cryptographic Software via
  Early-terminating Multiplications}. In
  \bibinfo{booktitle}{\emph{International Conference on Information Security
  and Cryptology}}. \bibinfo{pages}{176--192}.
\newblock


\bibitem[Guarnieri et~al\mbox{.}(2021)]%
        {sw-hw-contract}
\bibfield{author}{\bibinfo{person}{Marco Guarnieri}, \bibinfo{person}{Boris
  Köpf}, \bibinfo{person}{Jan Reineke}, {and} \bibinfo{person}{Pepe Vila}.}
  \bibinfo{year}{2021}\natexlab{}.
\newblock \showarticletitle{Hardware-Software Contracts for Secure
  Speculation}. In \bibinfo{booktitle}{\emph{2021 IEEE Symposium on Security
  and Privacy (SP)}}. \bibinfo{pages}{1868--1883}.
\newblock
\urldef\tempurl%
\url{https://doi.org/10.1109/SP40001.2021.00036}
\showDOI{\tempurl}


\bibitem[Haddad et~al\mbox{.}(2014)]%
        {jitter-rng}
\bibfield{author}{\bibinfo{person}{Patrick Haddad}, \bibinfo{person}{Yannick
  Teglia}, \bibinfo{person}{Florent Bernard}, {and} \bibinfo{person}{Viktor
  Fischer}.} \bibinfo{year}{2014}\natexlab{}.
\newblock \showarticletitle{On the Assumption of Mutual Independence of Jitter
  Realizations in P-TRNG Stochastic Models}. In \bibinfo{booktitle}{\emph{2014
  Design, Automation and Test in Europe Conference and Exhibition (DATE)}}.
  \bibinfo{pages}{1--6}.
\newblock
\urldef\tempurl%
\url{https://doi.org/10.7873/DATE.2014.052}
\showDOI{\tempurl}


\bibitem[Hamalainen et~al\mbox{.}(2006)]%
        {hamalainen2006design}
\bibfield{author}{\bibinfo{person}{Panu Hamalainen}, \bibinfo{person}{Timo
  Alho}, \bibinfo{person}{Marko Hannikainen}, {and} \bibinfo{person}{Timo~D
  Hamalainen}.} \bibinfo{year}{2006}\natexlab{}.
\newblock \showarticletitle{Design and Implementation of Low-Area and Low-Power
  AES Encryption Hardware Core}. In \bibinfo{booktitle}{\emph{9th EUROMICRO
  conference on digital system design (DSD'06)}}. IEEE,
  \bibinfo{pages}{577--583}.
\newblock


\bibitem[Hanna(2013)]%
        {website:hanna2013jasper}
\bibfield{author}{\bibinfo{person}{Ziyad Hanna}.}
  \bibinfo{year}{2013}\natexlab{}.
\newblock \bibinfo{title}{Jasper Case Study on Formally Verifying Secure
  On-Chip Datapaths. \url{https://www.deepchip.com/items/0524-03.html}}.
\newblock
\newblock
\urldef\tempurl%
\url{https://www.deepchip.com/items/0524-03.html}
\showURL{%
\tempurl}


\bibitem[Harris(2006)]%
        {harris2006rsa}
\bibfield{author}{\bibinfo{person}{Ben Harris}.}
  \bibinfo{year}{2006}\natexlab{}.
\newblock \bibinfo{booktitle}{\emph{RSA Key Exchange for the Secure Shell (SSH)
  Transport Layer Protocol}}.
\newblock \bibinfo{type}{{T}echnical {R}eport}.
\newblock


\bibitem[Heron(2009)]%
        {heron2009advanced}
\bibfield{author}{\bibinfo{person}{Simon Heron}.}
  \bibinfo{year}{2009}\natexlab{}.
\newblock \showarticletitle{Advanced Encryption Standard (AES)}.
\newblock \bibinfo{journal}{\emph{Network Security}} \bibinfo{volume}{2009},
  \bibinfo{number}{12} (\bibinfo{year}{2009}), \bibinfo{pages}{8--12}.
\newblock


\bibitem[Housley et~al\mbox{.}(1999)]%
        {housley1999internet}
\bibfield{author}{\bibinfo{person}{Russell Housley}, \bibinfo{person}{Warwick
  Ford}, \bibinfo{person}{William Polk}, {and} \bibinfo{person}{David Solo}.}
  \bibinfo{year}{1999}\natexlab{}.
\newblock \bibinfo{booktitle}{\emph{Internet X. 509 Public Key Infrastructure
  Certificate and CRL Profile
  \url{https://www.rfc-editor.org/rfc/rfc5280.html}}}.
\newblock \bibinfo{type}{{T}echnical {R}eport}.
\newblock


\bibitem[Hu et~al\mbox{.}(2021)]%
        {hu2021hardware}
\bibfield{author}{\bibinfo{person}{Wei Hu}, \bibinfo{person}{Armaiti
  Ardeshiricham}, {and} \bibinfo{person}{Ryan Kastner}.}
  \bibinfo{year}{2021}\natexlab{}.
\newblock \showarticletitle{Hardware information flow tracking}.
\newblock \bibinfo{journal}{\emph{ACM Computing Surveys (CSUR)}}
  \bibinfo{volume}{54}, \bibinfo{number}{4} (\bibinfo{year}{2021}),
  \bibinfo{pages}{1--39}.
\newblock


\bibitem[James~Wyllie(1944)]%
        {crypt-dict}
\bibfield{author}{\bibinfo{person}{Tony~Sale James~Wyllie}.}
  \bibinfo{year}{1944}\natexlab{}.
\newblock \showarticletitle{A Cryptographic Dictionary.}
\newblock \bibinfo{journal}{\emph{NR 4559, Historic Cryptographic Collection,
  Pre-World War I Through World War II, Record Group 457}}
  (\bibinfo{year}{1944}).
\newblock
\urldef\tempurl%
\url{https://www.codesandciphers.org.uk/documents/cryptdict/cryptix.htm}
\showURL{%
\tempurl}


\bibitem[Katz and Lindell(2020)]%
        {katz2020introduction}
\bibfield{author}{\bibinfo{person}{Jonathan Katz} {and} \bibinfo{person}{Yehuda
  Lindell}.} \bibinfo{year}{2020}\natexlab{}.
\newblock \bibinfo{booktitle}{\emph{Introduction to modern cryptography}}.
\newblock \bibinfo{publisher}{CRC press}.
\newblock


\bibitem[Khoo et~al\mbox{.}(2010)]%
        {khoo2010mixing}
\bibfield{author}{\bibinfo{person}{Yit~Phang Khoo},
  \bibinfo{person}{Bor-Yuh~Evan Chang}, {and} \bibinfo{person}{Jeffrey~S
  Foster}.} \bibinfo{year}{2010}\natexlab{}.
\newblock \showarticletitle{Mixing type checking and symbolic execution}.
\newblock \bibinfo{journal}{\emph{ACM Sigplan Notices}} \bibinfo{volume}{45},
  \bibinfo{number}{6} (\bibinfo{year}{2010}), \bibinfo{pages}{436--447}.
\newblock


\bibitem[Kim et~al\mbox{.}(2021)]%
        {kim2021massively}
\bibfield{author}{\bibinfo{person}{Kyungduk Kim}, \bibinfo{person}{Stefan
  Bittner}, \bibinfo{person}{Yongquan Zeng}, \bibinfo{person}{Stefano
  Guazzotti}, \bibinfo{person}{Ortwin Hess}, \bibinfo{person}{Qi~Jie Wang},
  {and} \bibinfo{person}{Hui Cao}.} \bibinfo{year}{2021}\natexlab{}.
\newblock \showarticletitle{Massively Parallel Ultrafast Random Bit Generation
  with a Chip-scale Laser}.
\newblock \bibinfo{journal}{\emph{Science}} \bibinfo{volume}{371},
  \bibinfo{number}{6532} (\bibinfo{year}{2021}), \bibinfo{pages}{948--952}.
\newblock


\bibitem[Kiniry et~al\mbox{.}(2018)]%
        {kiniry2018formally}
\bibfield{author}{\bibinfo{person}{Joseph~R Kiniry}, \bibinfo{person}{Daniel~M
  Zimmerman}, \bibinfo{person}{Robert Dockins}, {and} \bibinfo{person}{Rishiyur
  Nikhil}.} \bibinfo{year}{2018}\natexlab{}.
\newblock \showarticletitle{A Formally Verified Cryptographic Extension to a
  RISC-V Processor}.
\newblock \bibinfo{journal}{\emph{Computer Architecture Research with
  RISC-V--CARRV 2018}} (\bibinfo{year}{2018}).
\newblock


\bibitem[Kocher et~al\mbox{.}(1999)]%
        {kocher1999differential}
\bibfield{author}{\bibinfo{person}{Paul Kocher}, \bibinfo{person}{Joshua
  Jaffe}, {and} \bibinfo{person}{Benjamin Jun}.}
  \bibinfo{year}{1999}\natexlab{}.
\newblock \showarticletitle{Differential Power Analysis}. In
  \bibinfo{booktitle}{\emph{Advances in Cryptology—CRYPTO’99: 19th Annual
  International Cryptology Conference Santa Barbara, California, USA, August
  15--19, 1999 Proceedings 19}}. Springer, \bibinfo{pages}{388--397}.
\newblock


\bibitem[Kuan et~al\mbox{.}(2014)]%
        {7008853}
\bibfield{author}{\bibinfo{person}{Ting-Kuei Kuan}, \bibinfo{person}{Yu-Hsuan
  Chiang}, {and} \bibinfo{person}{Shen-Iuan Liu}.}
  \bibinfo{year}{2014}\natexlab{}.
\newblock \showarticletitle{A 0.43pJ/bit True Random Number Generator}. In
  \bibinfo{booktitle}{\emph{2014 IEEE Asian Solid-State Circuits Conference
  (A-SSCC)}}. \bibinfo{pages}{33--36}.
\newblock
\urldef\tempurl%
\url{https://doi.org/10.1109/ASSCC.2014.7008853}
\showDOI{\tempurl}


\bibitem[Lee et~al\mbox{.}(2020)]%
        {lee2020keystone}
\bibfield{author}{\bibinfo{person}{Dayeol Lee}, \bibinfo{person}{David
  Kohlbrenner}, \bibinfo{person}{Shweta Shinde}, \bibinfo{person}{Krste
  Asanovi{\'c}}, {and} \bibinfo{person}{Dawn Song}.}
  \bibinfo{year}{2020}\natexlab{}.
\newblock \showarticletitle{Keystone: An open framework for architecting
  trusted execution environments}. In \bibinfo{booktitle}{\emph{Proceedings of
  the Fifteenth European Conference on Computer Systems}}.
  \bibinfo{pages}{1--16}.
\newblock


\bibitem[Li et~al\mbox{.}(2021)]%
        {li2021cipherleaks}
\bibfield{author}{\bibinfo{person}{Mengyuan Li}, \bibinfo{person}{Yinqian
  Zhang}, \bibinfo{person}{Huibo Wang}, \bibinfo{person}{Kang Li}, {and}
  \bibinfo{person}{Yueqiang Cheng}.} \bibinfo{year}{2021}\natexlab{}.
\newblock \showarticletitle{CIPHERLEAKS: Breaking Constant-time Cryptography on
  AMD SEV via the Ciphertext Side Channel.}. In
  \bibinfo{booktitle}{\emph{USENIX Security Symposium}}.
  \bibinfo{pages}{717--732}.
\newblock


\bibitem[Liu et~al\mbox{.}(2022)]%
        {frequency-throttling}
\bibfield{author}{\bibinfo{person}{Chen Liu}, \bibinfo{person}{Abhishek
  Chakraborty}, \bibinfo{person}{Nikhil Chawla}, {and} \bibinfo{person}{Neer
  Roggel}.} \bibinfo{year}{2022}\natexlab{}.
\newblock \showarticletitle{Frequency Throttling Side-Channel Attack}. In
  \bibinfo{booktitle}{\emph{Proceedings of the 2022 ACM SIGSAC Conference on
  Computer and Communications Security}} (Los Angeles, CA, USA)
  \emph{(\bibinfo{series}{CCS '22})}. \bibinfo{publisher}{Association for
  Computing Machinery}, \bibinfo{address}{New York, NY, USA},
  \bibinfo{pages}{1977–1991}.
\newblock
\showISBNx{9781450394505}
\urldef\tempurl%
\url{https://doi.org/10.1145/3548606.3560682}
\showDOI{\tempurl}


\bibitem[Liu et~al\mbox{.}(2015)]%
        {prime-and-probe}
\bibfield{author}{\bibinfo{person}{Fangfei Liu}, \bibinfo{person}{Yuval Yarom},
  \bibinfo{person}{Qian Ge}, \bibinfo{person}{Gernot Heiser}, {and}
  \bibinfo{person}{Ruby~B. Lee}.} \bibinfo{year}{2015}\natexlab{}.
\newblock \showarticletitle{Last-Level Cache Side-Channel Attacks are
  Practical}. In \bibinfo{booktitle}{\emph{2015 IEEE Symposium on Security and
  Privacy}}. \bibinfo{pages}{605--622}.
\newblock
\urldef\tempurl%
\url{https://doi.org/10.1109/SP.2015.43}
\showDOI{\tempurl}


\bibitem[Loughlin et~al\mbox{.}(2021)]%
        {dolma}
\bibfield{author}{\bibinfo{person}{Kevin Loughlin}, \bibinfo{person}{Ian Neal},
  \bibinfo{person}{Jiacheng Ma}, \bibinfo{person}{Elisa Tsai},
  \bibinfo{person}{Ofir Weisse}, \bibinfo{person}{Satish Narayanasamy}, {and}
  \bibinfo{person}{Baris Kasikci}.} \bibinfo{year}{2021}\natexlab{}.
\newblock \showarticletitle{{DOLMA}: Securing Speculation with the Principle of
  Transient {Non-Observability}}. In \bibinfo{booktitle}{\emph{30th USENIX
  Security Symposium (USENIX Security 21)}}. \bibinfo{publisher}{USENIX
  Association}, \bibinfo{pages}{1397--1414}.
\newblock
\showISBNx{978-1-939133-24-3}
\urldef\tempurl%
\url{https://www.usenix.org/conference/usenixsecurity21/presentation/loughlin}
\showURL{%
\tempurl}


\bibitem[Mangard(2003)]%
        {mangard2003simple}
\bibfield{author}{\bibinfo{person}{Stefan Mangard}.}
  \bibinfo{year}{2003}\natexlab{}.
\newblock \showarticletitle{A Simple Power-analysis (SPA) Attack on
  Implementations of the AES Key Expansion}. In
  \bibinfo{booktitle}{\emph{Information Security and Cryptology—ICISC 2002:
  5th International Conference Seoul, Korea, November 28--29, 2002 Revised
  Papers 5}}. Springer, \bibinfo{pages}{343--358}.
\newblock


\bibitem[Maurer(1996)]%
        {maurer1996modelling}
\bibfield{author}{\bibinfo{person}{Ueli Maurer}.}
  \bibinfo{year}{1996}\natexlab{}.
\newblock \showarticletitle{Modelling a Public-Key Infrastructure}. In
  \bibinfo{booktitle}{\emph{European Symposium on Research in Computer
  Security}}. Springer, \bibinfo{pages}{325--350}.
\newblock


\bibitem[McKeen et~al\mbox{.}(2013)]%
        {mckeen2013innovative}
\bibfield{author}{\bibinfo{person}{Frank McKeen}, \bibinfo{person}{Ilya
  Alexandrovich}, \bibinfo{person}{Alex Berenzon}, \bibinfo{person}{Carlos~V
  Rozas}, \bibinfo{person}{Hisham Shafi}, \bibinfo{person}{Vedvyas Shanbhogue},
  {and} \bibinfo{person}{Uday~R Savagaonkar}.} \bibinfo{year}{2013}\natexlab{}.
\newblock \showarticletitle{Innovative Instructions and Software Model for
  Isolated Execution.}
\newblock \bibinfo{journal}{\emph{Hasp@ isca}} \bibinfo{volume}{10},
  \bibinfo{number}{1} (\bibinfo{year}{2013}).
\newblock


\bibitem[Mehibel and Hamadouche(2017)]%
        {mehibel2017new}
\bibfield{author}{\bibinfo{person}{Nissa Mehibel} {and}
  \bibinfo{person}{M'hamed Hamadouche}.} \bibinfo{year}{2017}\natexlab{}.
\newblock \showarticletitle{A New Approach of Elliptic Curve Diffie-Hellman Key
  Exchange}. In \bibinfo{booktitle}{\emph{2017 5th International Conference on
  Electrical Engineering-Boumerdes (ICEE-B)}}. IEEE, \bibinfo{pages}{1--6}.
\newblock


\bibitem[Myers(2011)]%
        {myers2011proving}
\bibfield{author}{\bibinfo{person}{Andrew Myers}.}
  \bibinfo{year}{2011}\natexlab{}.
\newblock \showarticletitle{Proving noninterference for a while-language using
  small-step operational semantics}.
\newblock  (\bibinfo{year}{2011}).
\newblock


\bibitem[Oberg et~al\mbox{.}(2014)]%
        {oberg2014leveraging}
\bibfield{author}{\bibinfo{person}{Jason Oberg}, \bibinfo{person}{Sarah
  Meiklejohn}, \bibinfo{person}{Timothy Sherwood}, {and} \bibinfo{person}{Ryan
  Kastner}.} \bibinfo{year}{2014}\natexlab{}.
\newblock \showarticletitle{Leveraging Gate-level Properties to Identify
  Hardware Timing Channels}.
\newblock \bibinfo{journal}{\emph{IEEE Transactions on Computer-Aided Design of
  Integrated Circuits and Systems}} \bibinfo{volume}{33}, \bibinfo{number}{9}
  (\bibinfo{year}{2014}), \bibinfo{pages}{1288--1301}.
\newblock


\bibitem[Petrie and Connelly(2000)]%
        {noiseic}
\bibfield{author}{\bibinfo{person}{C.S. Petrie} {and} \bibinfo{person}{J.A.
  Connelly}.} \bibinfo{year}{2000}\natexlab{}.
\newblock \showarticletitle{A Noise-based IC Random Number Generator for
  Applications in Cryptography}.
\newblock \bibinfo{journal}{\emph{IEEE Transactions on Circuits and Systems I:
  Fundamental Theory and Applications}} \bibinfo{volume}{47},
  \bibinfo{number}{5} (\bibinfo{year}{2000}), \bibinfo{pages}{615--621}.
\newblock
\urldef\tempurl%
\url{https://doi.org/10.1109/81.847868}
\showDOI{\tempurl}


\bibitem[Ravi et~al\mbox{.}(2020)]%
        {dropanddrop}
\bibfield{author}{\bibinfo{person}{Prasanna Ravi}, \bibinfo{person}{Shivam
  Bhasin}, \bibinfo{person}{Sujoy~Sinha Roy}, {and} \bibinfo{person}{Anupam
  Chattopadhyay}.} \bibinfo{year}{2020}\natexlab{}.
\newblock \bibinfo{title}{Drop by Drop you break the rock - Exploiting generic
  vulnerabilities in Lattice-based PKE/KEMs using EM-based Physical Attacks}.
\newblock \bibinfo{howpublished}{Cryptology ePrint Archive, Paper 2020/549}.
\newblock
\urldef\tempurl%
\url{https://eprint.iacr.org/2020/549}
\showURL{%
\tempurl}
\newblock
\shownote{\url{https://eprint.iacr.org/2020/549}}.


\bibitem[Ren et~al\mbox{.}(2017)]%
        {ren2017design}
\bibfield{author}{\bibinfo{person}{Ling Ren}, \bibinfo{person}{Christopher~W
  Fletcher}, \bibinfo{person}{Albert Kwon}, \bibinfo{person}{Marten Van~Dijk},
  {and} \bibinfo{person}{Srinivas Devadas}.} \bibinfo{year}{2017}\natexlab{}.
\newblock \showarticletitle{Design and Implementation of the Ascend Secure
  Processor}.
\newblock \bibinfo{journal}{\emph{IEEE Transactions on Dependable and Secure
  Computing}} \bibinfo{volume}{16}, \bibinfo{number}{2} (\bibinfo{year}{2017}),
  \bibinfo{pages}{204--216}.
\newblock


\bibitem[Rescorla(1999)]%
        {rescorla1999diffie}
\bibfield{author}{\bibinfo{person}{Eric Rescorla}.}
  \bibinfo{year}{1999}\natexlab{}.
\newblock \bibinfo{booktitle}{\emph{Diffie-hellman Key Agreement Method}}.
\newblock \bibinfo{type}{{T}echnical {R}eport}.
\newblock


\bibitem[Rodr{\'\i}guez~V{\'a}zquez et~al\mbox{.}(1991)]%
        {broadbandnoise}
\bibfield{author}{\bibinfo{person}{{\'A}ngel~Benito
  Rodr{\'\i}guez~V{\'a}zquez}, \bibinfo{person}{Manuel Delgado~Restituto},
  \bibinfo{person}{Servando~Carlos Espejo~Meana}, {and}
  \bibinfo{person}{Jos{\'e}~Luis Huertas~D{\'\i}az}.}
  \bibinfo{year}{1991}\natexlab{}.
\newblock \showarticletitle{A Switched-Capacitor Broadband Noise Generator for
  CMOS VLSI}.
\newblock \bibinfo{journal}{\emph{Electronics Letters, 27 (21), 1913-1915.}}
  (\bibinfo{year}{1991}).
\newblock


\bibitem[Rosulek({[n.\,d.]})]%
        {joyofcryptography}
\bibfield{author}{\bibinfo{person}{Mike Rosulek}.}
  \bibinfo{year}{[n.\,d.]}\natexlab{}.
\newblock \bibinfo{title}{The Joy of Cryptography}.
\newblock
\newblock
\urldef\tempurl%
\url{https://joyofcryptography.com}
\showURL{%
\tempurl}
\newblock
\shownote{\url{https://joyofcryptography.com}}.


\bibitem[Sabelfeld and Sands(2009)]%
        {sabelfeld2009declassification}
\bibfield{author}{\bibinfo{person}{Andrei Sabelfeld} {and}
  \bibinfo{person}{David Sands}.} \bibinfo{year}{2009}\natexlab{}.
\newblock \showarticletitle{Declassification: Dimensions and principles}.
\newblock \bibinfo{journal}{\emph{Journal of Computer Security}}
  \bibinfo{volume}{17}, \bibinfo{number}{5} (\bibinfo{year}{2009}),
  \bibinfo{pages}{517--548}.
\newblock


\bibitem[Schunter(2016)]%
        {schunter2016intel}
\bibfield{author}{\bibinfo{person}{Matthias Schunter}.}
  \bibinfo{year}{2016}\natexlab{}.
\newblock \showarticletitle{Intel Software Guard Extensions: Introduction and
  Open Research Challenges}. In \bibinfo{booktitle}{\emph{Proceedings of the
  2016 ACM Workshop on Software Protection}}. \bibinfo{pages}{1--1}.
\newblock


\bibitem[Sehatbakhsh et~al\mbox{.}(2020)]%
        {nscwithEM}
\bibfield{author}{\bibinfo{person}{Nader Sehatbakhsh},
  \bibinfo{person}{Baki~Berkay Yilmaz}, \bibinfo{person}{Alenka Zajic}, {and}
  \bibinfo{person}{Milos Prvulovic}.} \bibinfo{year}{2020}\natexlab{}.
\newblock \showarticletitle{A New Side-Channel Vulnerability on Modern
  Computers by Exploiting Electromagnetic Emanations from the Power Management
  Unit}. In \bibinfo{booktitle}{\emph{2020 IEEE International Symposium on High
  Performance Computer Architecture (HPCA)}}. \bibinfo{pages}{123--138}.
\newblock
\urldef\tempurl%
\url{https://doi.org/10.1109/HPCA47549.2020.00020}
\showDOI{\tempurl}


\bibitem[Solt et~al\mbox{.}(2022)]%
        {solt2022cellift}
\bibfield{author}{\bibinfo{person}{Flavien Solt}, \bibinfo{person}{Ben Gras},
  {and} \bibinfo{person}{Kaveh Razavi}.} \bibinfo{year}{2022}\natexlab{}.
\newblock \showarticletitle{{CellIFT}: Leveraging Cells for Scalable and
  Precise Dynamic Information Flow Tracking in {RTL}}. In
  \bibinfo{booktitle}{\emph{31st USENIX Security Symposium (USENIX Security
  22)}}. \bibinfo{pages}{2549--2566}.
\newblock


\bibitem[Suh et~al\mbox{.}(2003)]%
        {suh2003aegis}
\bibfield{author}{\bibinfo{person}{G~Edward Suh}, \bibinfo{person}{Dwaine
  Clarke}, \bibinfo{person}{Blaise Gassend}, \bibinfo{person}{Marten Van~Dijk},
  {and} \bibinfo{person}{Srinivas Devadas}.} \bibinfo{year}{2003}\natexlab{}.
\newblock \showarticletitle{AEGIS: Architecture for Tamper-Evident and
  Tamper-Resistant Processing}. In \bibinfo{booktitle}{\emph{ACM International
  Conference on Supercomputing 25th Anniversary Volume}}.
  \bibinfo{pages}{357--368}.
\newblock


\bibitem[Suh et~al\mbox{.}(2005)]%
        {suh2005design}
\bibfield{author}{\bibinfo{person}{G~Edward Suh}, \bibinfo{person}{Charles~W
  O'Donnell}, \bibinfo{person}{Ishan Sachdev}, {and} \bibinfo{person}{Srinivas
  Devadas}.} \bibinfo{year}{2005}\natexlab{}.
\newblock \showarticletitle{Design and Implementation of the AEGIS Single-Chip
  Secure Processor using Physical Random Functions}. In
  \bibinfo{booktitle}{\emph{32nd International Symposium on Computer
  Architecture (ISCA'05)}}. IEEE, \bibinfo{pages}{25--36}.
\newblock


\bibitem[Tiwari et~al\mbox{.}(2009)]%
        {tiwari2009complete}
\bibfield{author}{\bibinfo{person}{Mohit Tiwari}, \bibinfo{person}{Hassan~MG
  Wassel}, \bibinfo{person}{Bita Mazloom}, \bibinfo{person}{Shashidhar Mysore},
  \bibinfo{person}{Frederic~T Chong}, {and} \bibinfo{person}{Timothy
  Sherwood}.} \bibinfo{year}{2009}\natexlab{}.
\newblock \showarticletitle{Complete Information Flow Tracking from the Gates
  Up}. In \bibinfo{booktitle}{\emph{Proceedings of the 14th international
  conference on Architectural support for programming languages and operating
  systems}}. \bibinfo{pages}{109--120}.
\newblock


\bibitem[Volpano et~al\mbox{.}(1996)]%
        {volpano1996sound}
\bibfield{author}{\bibinfo{person}{Dennis Volpano}, \bibinfo{person}{Cynthia
  Irvine}, {and} \bibinfo{person}{Geoffrey Smith}.}
  \bibinfo{year}{1996}\natexlab{}.
\newblock \showarticletitle{A sound type system for secure flow analysis}.
\newblock \bibinfo{journal}{\emph{Journal of computer security}}
  \bibinfo{volume}{4}, \bibinfo{number}{2-3} (\bibinfo{year}{1996}),
  \bibinfo{pages}{167--187}.
\newblock


\bibitem[Wang et~al\mbox{.}(2022)]%
        {hertzbleed}
\bibfield{author}{\bibinfo{person}{Yingchen Wang}, \bibinfo{person}{Riccardo
  Paccagnella}, \bibinfo{person}{Elizabeth~Tang He}, \bibinfo{person}{Hovav
  Shacham}, \bibinfo{person}{Christopher~W. Fletcher}, {and}
  \bibinfo{person}{David Kohlbrenner}.} \bibinfo{year}{2022}\natexlab{}.
\newblock \showarticletitle{Hertzbleed: Turning Power {Side-Channel} Attacks
  Into Remote Timing Attacks on x86}. In \bibinfo{booktitle}{\emph{31st USENIX
  Security Symposium (USENIX Security 22)}}. \bibinfo{publisher}{USENIX
  Association}, \bibinfo{address}{Boston, MA}, \bibinfo{pages}{679--697}.
\newblock
\showISBNx{978-1-939133-31-1}
\urldef\tempurl%
\url{https://www.usenix.org/conference/usenixsecurity22/presentation/wang-yingchen}
\showURL{%
\tempurl}


\bibitem[Wei and Guo(2009)]%
        {Wei:09}
\bibfield{author}{\bibinfo{person}{Wei Wei} {and} \bibinfo{person}{Hong Guo}.}
  \bibinfo{year}{2009}\natexlab{}.
\newblock \showarticletitle{Bias-free true random-number generator}.
\newblock \bibinfo{journal}{\emph{Opt. Lett.}} \bibinfo{volume}{34},
  \bibinfo{number}{12} (\bibinfo{date}{Jun} \bibinfo{year}{2009}),
  \bibinfo{pages}{1876--1878}.
\newblock
\urldef\tempurl%
\url{https://doi.org/10.1364/OL.34.001876}
\showDOI{\tempurl}


\bibitem[Weise(2001)]%
        {weise2001public}
\bibfield{author}{\bibinfo{person}{Joel Weise}.}
  \bibinfo{year}{2001}\natexlab{}.
\newblock \showarticletitle{Public Key Infrastructure Overview}.
\newblock \bibinfo{journal}{\emph{Sun BluePrints OnLine, August}}
  (\bibinfo{year}{2001}), \bibinfo{pages}{1--27}.
\newblock


\bibitem[Wong(2005)]%
        {wong2005timing}
\bibfield{author}{\bibinfo{person}{Wing~H Wong}.}
  \bibinfo{year}{2005}\natexlab{}.
\newblock \showarticletitle{Timing Attacks on RSA: Revealing Your Secrets
  through the Fourth Dimension}.
\newblock \bibinfo{journal}{\emph{XRDS: Crossroads, The ACM Magazine for
  Students}} \bibinfo{volume}{11}, \bibinfo{number}{3} (\bibinfo{year}{2005}),
  \bibinfo{pages}{5--5}.
\newblock


\bibitem[Yasuda et~al\mbox{.}(2004)]%
        {1317067}
\bibfield{author}{\bibinfo{person}{S. Yasuda}, \bibinfo{person}{H. Satake},
  \bibinfo{person}{T. Tanamoto}, \bibinfo{person}{R. Ohba}, \bibinfo{person}{K.
  Uchida}, {and} \bibinfo{person}{S. Fujita}.} \bibinfo{year}{2004}\natexlab{}.
\newblock \showarticletitle{Physical Random Number Generator Based on MOS
  Structure after Soft Breakdown}.
\newblock \bibinfo{journal}{\emph{IEEE Journal of Solid-State Circuits}}
  \bibinfo{volume}{39}, \bibinfo{number}{8} (\bibinfo{year}{2004}),
  \bibinfo{pages}{1375--1377}.
\newblock
\urldef\tempurl%
\url{https://doi.org/10.1109/JSSC.2004.831480}
\showDOI{\tempurl}


\bibitem[Yu et~al\mbox{.}({[n.\,d.]})]%
        {yu_data_2019}
\bibfield{author}{\bibinfo{person}{Jiyong Yu}, \bibinfo{person}{Lucas Hsiung},
  \bibinfo{person}{Mohamad El'Hajj}, {and} \bibinfo{person}{Christopher~W.
  Fletcher}.} \bibinfo{year}{[n.\,d.]}\natexlab{}.
\newblock \showarticletitle{Data Oblivious {ISA} Extensions for Side
  Channel-Resistant and High Performance Computing}. In
  \bibinfo{booktitle}{\emph{Proceedings 2019 Network and Distributed System
  Security Symposium}} (San Diego, {CA}, 2019). \bibinfo{publisher}{Internet
  Society}.
\newblock
\showISBNx{978-1-891562-55-6}
\urldef\tempurl%
\url{https://doi.org/10.14722/ndss.2019.23061}
\showDOI{\tempurl}


\bibitem[Yu et~al\mbox{.}(2020)]%
        {sdo}
\bibfield{author}{\bibinfo{person}{Jiyong Yu}, \bibinfo{person}{Namrata
  Mantri}, \bibinfo{person}{Josep Torrellas}, \bibinfo{person}{Adam Morrison},
  {and} \bibinfo{person}{Christopher~W. Fletcher}.}
  \bibinfo{year}{2020}\natexlab{}.
\newblock \showarticletitle{Speculative Data-Oblivious Execution: Mobilizing
  Safe Prediction For Safe and Efficient Speculative Execution}. In
  \bibinfo{booktitle}{\emph{2020 ACM/IEEE 47th Annual International Symposium
  on Computer Architecture (ISCA)}}. \bibinfo{pages}{707--720}.
\newblock
\urldef\tempurl%
\url{https://doi.org/10.1109/ISCA45697.2020.00064}
\showDOI{\tempurl}


\bibitem[Yu et~al\mbox{.}(2019)]%
        {stt}
\bibfield{author}{\bibinfo{person}{Jiyong Yu}, \bibinfo{person}{Mengjia Yan},
  \bibinfo{person}{Artem Khyzha}, \bibinfo{person}{Adam Morrison},
  \bibinfo{person}{Josep Torrellas}, {and} \bibinfo{person}{Christopher~W.
  Fletcher}.} \bibinfo{year}{2019}\natexlab{}.
\newblock \showarticletitle{Speculative Taint Tracking (STT): A Comprehensive
  Protection for Speculatively Accessed Data}. In
  \bibinfo{booktitle}{\emph{Proceedings of the 52nd Annual IEEE/ACM
  International Symposium on Microarchitecture}} (Columbus, OH, USA)
  \emph{(\bibinfo{series}{MICRO '52})}. \bibinfo{publisher}{Association for
  Computing Machinery}, \bibinfo{address}{New York, NY, USA},
  \bibinfo{pages}{954–968}.
\newblock
\showISBNx{9781450369381}
\urldef\tempurl%
\url{https://doi.org/10.1145/3352460.3358274}
\showDOI{\tempurl}


\bibitem[Zhang et~al\mbox{.}(2015)]%
        {zhang2015hardware}
\bibfield{author}{\bibinfo{person}{Danfeng Zhang}, \bibinfo{person}{Yao Wang},
  \bibinfo{person}{G~Edward Suh}, {and} \bibinfo{person}{Andrew~C Myers}.}
  \bibinfo{year}{2015}\natexlab{}.
\newblock \showarticletitle{A Hardware Design Language for Timing-sensitive
  Information-flow Security}.
\newblock \bibinfo{journal}{\emph{Acm Sigplan Notices}} \bibinfo{volume}{50},
  \bibinfo{number}{4} (\bibinfo{year}{2015}), \bibinfo{pages}{503--516}.
\newblock


\bibitem[Çiçek and Dündar(2013)]%
        {chaoticjitter}
\bibfield{author}{\bibinfo{person}{İhsan Çiçek} {and}
  \bibinfo{person}{Günhan Dündar}.} \bibinfo{year}{2013}\natexlab{}.
\newblock \showarticletitle{A Chaos-based Integrated Jitter Booster Circuit for
  True Random Number Generators}. In \bibinfo{booktitle}{\emph{2013 European
  Conference on Circuit Theory and Design (ECCTD)}}. \bibinfo{pages}{1--4}.
\newblock
\urldef\tempurl%
\url{https://doi.org/10.1109/ECCTD.2013.6662257}
\showDOI{\tempurl}


\end{thebibliography}

\clearpage
\appendix
\section{Formal Statement of Cryptographic Primitives}
\label{appx:crypto-def}
Will start by formally defining the notion of security against chosen ciphertext attacks. For a give encryption scheme $\pi$ we define the following experiment $Exp_{\mathcal{A}, \pi}^{CCA}(n+s)$ parameterized by the security parameters $n,s$ and  an arbitrary $(n+s)$-polynomially bounded attacker $\mathcal{A}$:

\begin{enumerate}
    \item $\mathcal{A}$ is given oracle access to $enc_k(\cdot)$ and $dec_k(\cdot)$. It eventually outputs two messages $m_0, m_1$ of size $n$.
    \item A uniform bit $b\in \{0, 1\}$ is selected and the challenge ciphertext $c \leftarrow enc_k(m_b)$ is computed and given to $\mathcal{A}$.
    \item $\mathcal{A}$ continues to have oracle access to $enc_k(\cdot)$ and $dec_k(\cdot)$ except when $dec_k(c)$ is queried the oracle returns $\bot$. Eventually, $\mathcal{A}$ outputs a bit $b'$.
    \item The result of the experiment is 1 when $b = b'$ and 0 otherwise. 
\end{enumerate}
We can state the security requirement of the SE system formally as follows,

\begin{definition}[CCA Security]\label{def:cca-security}
    An encryption scheme $\pi$ is CCA-secure in the following sense
    \begin{align*}
        Pr[Exp_{\mathcal{A}, \pi}^{CCA}(n+s) = 1] \leq \frac{1}{2} + negl(n+s)
    \end{align*}
    Where $negl$ is a function from the natural numbers to the non-negative real numbers where for every polynomial $p$ there is an $X$ such that for all $x > X$ it holds that $negl(x) < \frac{1}{p(x)}$.
\end{definition}

\begin{theorem}[SE CCA Security]\label{thm:se-security}
   The $SE$ encryption scheme defined in \S\ref{sub:formalization-isa} is CCA-secure. 
\end{theorem}
\begin{proof}
    The proof is a straightforward reduction to the properties of the strong pseudorandom permutation used by the scheme as long as $u$ is parameterized by $s$ and freshly generated for each query. The detailed proof is presented in introductory textbooks, such as \cite[Ch. 9.4]{joyofcryptography}.
\end{proof}

\section{Full Proof of Soundness}
\label{appx:proofs}
\subsection{Proof of Lemma \ref{lma:single-command}}
\label{appx:proofs-lemma}
\begin{proof}
The proof follows by structural induction on the derivation $\langle p, \sigma_1 \rangle \xrightarrow{t(\sigma)} \langle \mathsf{p_1'}, \sigma_1' \rangle$. First, note that the programs $p_1'$ and $p_2'$ never diverge in this semantics, which isn't true for small-step semantics generally. For the rules \verb|CMOV-T|, \verb|CMOV-F|, \verb|BOP|, and \verb|ENC| all step into the same configuration of $\langle \mathsf{skip}, \sigma \rangle$ thus if $p$ is one of these, we will always have $p_1' = p_2' = \mathsf{skip}$. In the case of the base case \verb|SEQ| rule, the step ends with the configuration $\langle q, \sigma \rangle$ in which case $p_1' = p_2' = q$ since $p = \mathsf{skip; q}$. In the main \verb|SEQ| case, we end with the configuration $\langle \mathsf{skip; p} \rangle$

Now suppose, 

\begin{enumerate}
\item The small-step taken is $\mathsf{CMOV-T}$

      Under $\sigma_1$
     \begin{mathpar}
        \inferrule[]{
            \snr{r_1}{\sigma_1}{[c_1]}\\
            \snr{r_3}{\sigma_1}{[c_3]}\\
            \snr{keyReg}{\sigma_1}{k}\\
            decrypt(c_1, k) = b\\
            decrypt(c_3, k) = m\\
            u \sim uniform(s) \\
            encrypt(m||u, k) = [c_5]\\
        }{\sn{\mathsf{if}~r_1: r_2 \leftarrow r_3~\mathsf{else}~r_2 \leftarrow r_4}{~\sigma_1}{skip}{\sigma_1[[c_5]/r_2]}{2}} 
    \end{mathpar}

 and under $\sigma_2$

\begin{mathpar}
      \inferrule[]{
          \snr{r_1}{\sigma_2}{[c_1']}\\
          \snr{r_3}{\sigma_2}{[c_3']}\\
          \snr{keyReg}{\sigma_2}{k'}\\
          decrypt(c_1, k') = b'\\
          decrypt(c_3, k') = m'\\
          u' \sim uniform(s) \\
          encrypt(m'||u', k') = [c_5']\\
      }{\sn{\mathsf{if}~r_1: r_2 \leftarrow r_3~\mathsf{else}~r_2 \leftarrow r_4}{~\sigma_2}{skip}{\sigma_2[[c_5']/r_2]}{2}} 
  \end{mathpar}

  and the typing rule CMOV is applied

  \begin{mathpar}
    \inferrule[]{
        \tp {r_1} \pu \\
        \tp {r_2} \pu \\
        \tp {r_3} \pu \\
        \tp {r_4} \pu 
      }{\tp  {if~r_1:r_2 \leftarrow r_3~else~r_2 \leftarrow r_4} \pud}
  \end{mathpar}

Since $\sigma_1 \approx_l \sigma_2$ and $\tp {r_1} \pu$ and $\tp {r_3} \pu$ by the definition of $\G$, we have $[c_1] = [c_1']$ and $[c_3] = [c_3']$. Since, $\tp k \pr$, we then could consider two cases both of which resolve in the same way ultimately:

\begin{itemize}
    \item \textit{$k = k'$}. Then we get $b = b'$ and $m = m'$ since the decryption equation holds for $decrypt$.
    \item \textit{$k \neq k'$}. Then $b \neq b'$ and $m \neq m'$. But since $\tp b \pr$ and $\tp m \pr$ and same for $b', m'$ and there is now flow to a public register, this is irrelevant for the low-equivalence of the states.
\end{itemize}

In either case, we have $\tp {[c_5]} \pu$ and $\tp{[c_5']} \pu$ but then $[c_5] \approx [c_5']$ by Definition \ref{def:equiv}. Since in both cases only $r_2$ is updated by an equivalent value and $\sigma_1 \approx_l \sigma_2$, we can conclude $\G \vdash \sigma_1' \approx_l \sigma_2'$.

\item The rules \verb|CMOV-F|, \verb|BOP|, \verb|ENC| follow the same argument as \verb|CMOV-T|.

\item Suppose the step applied is the base \verb|SEQ| rule. 

Under $\sigma_1$:
\begin{mathpar}
    \inferrule[seq]{\ }{\sn{\mathsf{skip; q}}{\sigma_1}{q}{\sigma_1}{1}}
\end{mathpar}

and under $\sigma_2$:
\begin{mathpar}
    \inferrule[seq]{\ }{\sn{\mathsf{skip; q}}{\sigma_2}{q}{\sigma_2}{1}}
\end{mathpar}

And the typing rule \verb|SEQ| is applied
\begin{mathpar}
  \inferrule[seq]{
    \tp {p_1} \ell'~prog \\
     \tp {p_2} \ell''~prog\\
     \ell = \ell' \sqcup \ell''
  }{\tp {p_1; p_2} \ell~prog}
\end{mathpar}

Since $\G \vdash \sigma_1 \approx_l \sigma_2$ and $\sigma_1 = \sigma_1' \land \sigma_2 = \sigma_2'$, we have $\G \vdash \sigma_1' \approx_l \sigma_2'$ trivially.

\item Finally, suppose the step taken is the second \verb|SEQ| rule

Under $\sigma_1$:
\begin{mathpar}
   \inferrule[]{\sn{\mathsf{c}}{\sigma_1}{\mathsf{skip}}{\sigma_1'}{1}}{\sn{\mathsf{c;p}}{\sigma_1}{skip;p}{\sigma_1'}{1}}
\end{mathpar}

and under $\sigma_2$:
\begin{mathpar}
   \inferrule[]{\sn{\mathsf{c}}{\sigma_2}{\mathsf{skip}}{\sigma_2'}{1}}{\sn{\mathsf{c;p}}{\sigma_2}{skip;p}{\sigma_2'}{1}}
\end{mathpar}

And the typing rule $SEQ$ is applied:

\begin{mathpar}
  \inferrule[seq]{
    \tp {p_1} \ell'~prog \\
     \tp {p_2} \ell''~prog\\
     \ell = \ell' \sqcup \ell''
  }{\tp {p_1; p_2} \ell~prog}
\end{mathpar}

By the inductive hypothesis on $\sn{\mathsf{c}}{\sigma_1}{\mathsf{skip}}{\sigma_1'}{1}$, we already get $\sigma'_1 \approx_l \sigma_2'$.

\end{enumerate}

The timing property $t(\sigma_1) = t(\sigma_2)$ is immediate from assumption (2) and the restriction on $t$ by Eq. \ref{eq:time-invariance}.
\end{proof}

\subsection{Proof of Theorem \ref{thm:soundness}}
\label{appx:proofs-thm}
We will give the usual definition of $\rightarrow_*$ adjusted to account for the timing property. 

\begin{definition}[Multi-step]
The multi-step function $\xrightarrow{n}_*$ is the transitive and reflexive closure of $\xrightarrow{t(\sigma)}$. Inductively,
\begin{mathpar}
    \inferrule[]{\ }{\langle p, \sigma \rangle \xrightarrow{0}_* \langle p, \sigma \rangle}\and
    \inferrule[]{\sn{p}{\sigma}{p'}{\sigma'}{1}\\
                \langle p', \sigma \rangle \xrightarrow{n}_* \langle p'', \sigma'' \rangle}{\langle p, \sigma \rangle \xrightarrow{t(\sigma) + n}_* \langle p'', \sigma'' \rangle}\and
\end{mathpar}
\end{definition}

Now we can prove the theorem as follows. 
We will first state a preservation property here:

\begin{lemma}[Preservation]
\label{lma:preservation}
  If 
\begin{enumerate}
  \item $p$ be an SE program
  \item $\tp p \pu~prog$
  \item $\sigma$ be an initial state 
  \item $\sn{p}{\sigma}{p'}{\sigma'}{}$
\end{enumerate}

Then we have that $p'$ is also well-typed: $\tp {p'} \pu~prog$
\end{lemma}

\begin{proof}
The proof follows by induction on the derivation of $p'$.  
  
\end{proof}

\begin{proof}
The proof follows by induction on the number of steps taken by $\rightarrow^*$. The base case of a single step is immediate from Lemma \ref{lma:single-command}. Now assume it holds for $k$ steps, then 
\begin{enumerate}
    \item $\tp {p_k} \ell$ for the program after $k$ steps under $\sigma_1$
    \item $\tp {p'_k} \ell$ for the program after $k$ steps under $\sigma_2$ 
    \item $\G \vdash \sigma_k \approx_l \sigma'_k$ the states after $k$ steps
\end{enumerate}

By our argument in Appendix \ref{appx:proofs-lemma}, we know $\mathsf{p_k} = \mathsf{p'_k}$ since we don't have diverging programs. Then we also have,

\begin{enumerate}
    \item $\sn{p_k}{\sigma_k}{q}{\sigma'}{1}$
    \item $\sn{p'_k}{\sigma'_k}{q}{\sigma''}{1}$
\end{enumerate}

This is exactly what is needed by Lemma \ref{lma:single-command} to take another step given the type preservation by Lemma \ref{lma:preservation}. Thus we can conclude $\sigma_1' \approx_l \sigma_2'$.

The timing property follows directly from Lemma \ref{lma:single-command} for the base case as well. For the inductive case, note that $\mathsf{p_k} = \mathsf{p'_k}$ and that $\sigma_k \approx_l \sigma'_k$ at each step. Thus at each step $t(\sigma_k) = t(\sigma'_k)$ by Eq. \ref{eq:time-invariance}. Since $\rightarrow_*$ simply sums the trace of natural numbers generated by $t$, we have that $n = m$.
\end{proof}

\lstset{
 columns=fixed,       
 numbers=left,                                
 numberstyle=\tiny\color{gray},                    
 frame=none,                                    
 backgroundcolor=\color[RGB]{245,245,244},            
 keywordstyle=\color[RGB]{40,40,255},              
 numberstyle=\footnotesize\color{darkgray},           
 commentstyle=\it\color[RGB]{0,96,96},               
 stringstyle=\rmfamily\slshape\color[RGB]{128,0,0},  
 showstringspaces=false,                              
 language=verilog,                                       
}

\section{A Binary Search Attack on a Ciphertext with Reused Entropy Bits}
\label{appx:attack-example}
\begin{lstlisting}[caption={A simple bruteforce exploit for an operator that does not refresh its entropy.}]
enc_bool true_predicate = (enc_int)1 < 2;
enc_bool false_predicate = (enc_int)2 < 1;
int guess = MAX_INT / |2|;
for (int i=30; i>0; i--) {
    enc_bool cond = (guess < secret);
    //If ciphertexts are not equal
    if ((int)cond != (int)true_predicate) {
        guess -= (1 << i);
    }
    //Else if ciphertexts are equal
    else {
        guess += (1 << i);
    }
}
\end{lstlisting}

The code snippet demonstrates how a poorly designed \verb|lt| ($<$) operator for ciphertexts can lead to a simple brute force exploit that extracts the secret value of the ciphertext without using the key. Specifically, the operator here does not refresh the entropy bits used during the encryption process. The attacker first generates a secure ciphertext representing true and false. Then it simply starts with the middle representable integer as a guess and performs a binary search over all representable integers in the system each time comparing with the true and false ciphertexts generated in the first step. For a 32-bit secret integer, the attacker can extract the secret value in 31 iterations.

\section{Example Code for Vulnerable Designs}
\subsection{Value-dependent Timing Multiplier}
\label{appx:example_mul}
\begin{lstlisting}
input clk;
reg [63:0] reg_a, reg_b;  //two operands
reg [127:0] o;            //product
reg [5:0] counter;        //shift counts
reg finish; 

initial begin
    o = 0;
    finish = 0;
    counter = 0;
end

always @(posedge clk) begin
    if (reg_a == 0 || reg_b == 0) begin
        finish <= 1;
    end
    else begin             //shift and add
        if (reg_b[0] == 1)
            o <= o + reg_a << counter;
        reg_b <= reg_b >> 1;
        counter <= counter + 1;
    end
end

\end{lstlisting}

\subsection{Value-dependent Timing RSA}
\label{appx:example_rsa}
\begin{lstlisting}
input clk;
reg finish;

reg [127:0] o, o_next, d_leftover, n;
initial begin
    o_next = c; // c is the ciphertext
    o = 1;
    finish = 0;
    d_leftover = d; // d is decryption key
    n = N; // N is the modulus
end

always @(posedge clk) begin
    if(d_leftover[0] == 1) begin
        o <= (o * o_next) % n;
    end

    o_next <= (o_next*o_next) % n;
    d_leftover <= (d_leftover >> 1);
end

assign finish = (d_leftover == 0);
\end{lstlisting}
\end{document}